\documentclass[a4paper,reqno]{amsart}
\numberwithin{equation}{section}
\usepackage{amsthm,amsmath,amssymb}
\usepackage{color}
\usepackage[dvips]{hyperref}
\usepackage{stmaryrd}

\newcommand{\blue}[1]{\textcolor{blue}{\textbf{#1}}}
\newcommand{\mymnote}[1]{\marginpar{\raggedright\tiny\em
$\!\!\!\!\!\!\,\bullet$ \blue{#1}}}	  %Marginal Note
\renewcommand{\mymnote}[1]{}	    %no Marginal Note

\usepackage[T1]{fontenc}

\makeatletter
  \def\moverlay{\mathpalette\mov@rlay}
  \def\mov@rlay#1#2{\leavevmode\vtop{%
     \baselineskip\z@skip \lineskiplimit-\maxdimen
     \ialign{\hfil$#1##$\hfil\cr#2\crcr}}}
\makeatother

\newcommand{\squareTME}{\moverlay{\square\cr {\scriptscriptstyle \mathrm T}}}

\newcommand{\Co}{\mathbb{C}}		%the complex numbers
\renewcommand{\Re}{{\mathfrak{Re}}}	%the real part
\renewcommand{\Im}{{\mathfrak{Im}}}	%the imaginary part
\newcommand{\half}{\frac{1}{2}}		%half
\newcommand{\eps}{\epsilon}

\renewcommand{\i}{\mathrm{i}}

\newcommand{\inv}[1]{}

\newcommand{\fraka}{\mathfrak{a}}
\newcommand{\frakb}{\mathfrak{b}}
\newcommand{\frakc}{\mathfrak{c}}
\newcommand{\frakd}{\mathfrak{d}}

\newcommand{\NPl}{l}
\newcommand{\NPn}{n}
\newcommand{\NPm}{m}
\newcommand{\NPmbar}{\bar{m}}

\newcommand{\NPrho}{\rho}
\newcommand{\NPtau}{\tau}

\DeclareMathOperator*{\diag}{diag}

\DeclareMathOperator*{\tho}{\text{\th}}
\DeclareMathOperator*{\edt}{\text{\dh}}

\newcommand{\hnabla}{\widehat{\nabla}}

\DeclareMathOperator*{\SO}{SO}

\newcommand{\EE}{\mathcal E}

\theoremstyle{plain}
\newtheorem{thm}{Theorem}[section]
\newtheorem{cor}[thm]{Corollary}
\newtheorem{lemma}[thm]{Lemma}

\newtheorem{remark}[thm]{Remark}

\title{Linearized gravity and gauge conditions}

\author[S. Aksteiner]{Steffen Aksteiner}
\email{steffen@zarm.uni-bremen.de}
\address{QUEST, Leibniz University Hannover, Welfengarten 1, D-30167
  Hannover, Germany%}
\and
ZARM, University of Bremen, Am Fallturm 1, D-28359 Bremen, Germany}

\author[L. Andersson]{Lars Andersson}%${}^\dagger$}
\email{laan@aei.mpg.de}
\address{Albert Einstein Institute, Am M\"uhlenberg 1, D-14476 Potsdam,
  Germany}

\usepackage{comment}
\excludecomment{longversion}
\excludecomment{wrongversion}

\begin{document}

\date{September 28, 2010}

\begin{abstract}
In this paper we consider the field
equations for linearized gravity and other
integer spin fields on the Kerr spacetime, and more generally on
spacetimes
of Petrov type D. We give a derivation, using the GHP formalism, of decoupled
field equations for the linearized Weyl scalars for all spin weights and
identify the gauge source functions occuring in these.
For the spin weight 0 Weyl scalar, imposing a generalized
harmonic
coordinate gauge yields a generalization of the Regge-Wheeler equation.
Specializing to the Schwarzschild case,
we derive the gauge invariant Regge-Wheeler and Zerilli equation directly
from the equation for the spin 0 scalar.
\end{abstract}

\maketitle

\tableofcontents

\section{Introduction} \label{sec:introduction}
In his 1965 paper \cite{penrose:1965}
Penrose showed that, at least formally, all solutions of
the massless spin $s$ field
equation on Minkowski space can be
obtained from solutions of the spin $(s-\half)$
equations. By repeating this
process $2s$ times, one finds that solutions of to the spin $0$ equation
(i.e. the free scalar wave equation) are potentials for the massless
spin $s$ field. This is an example of the spin-raising and -lowering
transformations discussed in more detail in \cite[section 6.4]{PR:II}, see
also \cite[section 6.7]{PR:II}.

The following special case of the construction is relevant here.
Given a 2 index
Killing spinor, i.e. a spinor $K_{AB}$ solving the equation
$$
\nabla_{(A}{}^{A'} K_{BC)} = 0 ,
$$
and a solution $\phi_{ABC\dots D} = \phi_{(ABC\dots D)}$
of the spin $s$ zero rest-mass equation
$$
\nabla^{AA'} \phi_{ABC\dots D} = 0 ,
$$
the spin-lowered field $\hat \phi_{C\dots D} = \phi_{ABC\dots D} K^{AB}$
is a solution of the spin $(s-1)$ zero rest-mass equation.

The analysis of linear field equations on Minkowski space is a key step in
the proof of the non-linear stability of Minkowski space
\cite{ChristodoulouKlainerman:LinearFields, christo:klain:book}.
For the case of Minkowski space, the linearized Bianchi equation is precisely
the massless spin 2 equation. Thus, the discussion of Penrose shows that the
linearized stability problem for the
Einstein equation on Minkowski space can essentially
be reduced to a study of the scalar wave equation.

The problem of non-linear stability of Minkowski space, solved in
\cite{christo:klain:book}, see also \cite{lindblad:rodnianski:global},
can be viewed as a
warm-up for the \emph{black hole stability problem},
i.e. the problem of
proving the non-linear stability of the Kerr spacetime in the class of
asymptotically flat vacuum spacetimes, see e.g.
\cite{andersson:blue:kerr,Dafermos:2008en,finster:etal:bams:2009},
and references therein.
The black hole stability problem
adds several levels of
difficulty over the problem of stability of Minkowski space.
The aspect which we shall focus on in the present discussion is
that the background spacetime (or more properly stated, the asymptotic state for
the evolution)
fails to be conformally flat and consequently the relation between
the spin $0$ wave equation, viewed as a model problem for the full non-linear
stability problem,
and the equations of linearized gravity fails to be as
close as on the Minkowski background.

In particular, the equation for linearized gravity (i.e. the linearized
Bianchi system) on a non-conformally flat vacuum background is not the spin-2
system, but has a non-trivial right hand side. In fact,
the massless spin $2$ equation in a spacetime of Petrov type D has only trivial
solutions, cf. \cite{Buchdahl:1958xv}. It follows that
the spin-lowering and -raising transformations
cannot be applied directly to Maxwell or linearized gravity
on a background spacetime which is not conformally
flat.

However, if we consider the field equations of integer spin
on a vacuum type D spacetime, of which Kerr is a special case,
then it turns out that an analogue of the spin-lowering transformation
does yield useful results,
even for the equations of linearized gravity. This idea was
alluded to already in the paper of Jeffryes, see
\cite[p. 340]{jeffryes:1984}.
Although not stated as explicitly,
related ideas play an important
role in the work
of Fackerell and Crossman \cite{crossman:fackerell:1982,fackerell:1982}.

From the point of view of the black hole stability problem, the two-parameter
Kerr family of rotating black hole solutions are the vacuum type D spacetimes
of most immediate interest. However, most of
the results in this paper are valid for general vacuum type D spacetimes. These
include
the Kerr-Taub-NUT spacetimes, see \cite[section 21.1]{exact:solutions:book},
see also \cite{edgar:etal:2009}.

\subsection{Spin-lowering the Maxwell field}
Recall\footnote{Here and below we
use the conventions and notations of the Newman-Penrose (NP) and
Geroch-Held-Penrose (GHP) null-tetrad based formalisms, and the associated
two-spinor formalism, see \cite{newman:penrose:1962,GHP,PR:I,PR:II} and
section \ref{sec:preliminaries} below.} \cite{walker:penrose:1970} that any vacuum type D spacetime admits a Killing spinor of the form
\begin{equation}\label{eq:killing-spinor-intro}
K_{AB} = \Psi_2^{-1/3} o_{(A} \iota_{B)} .
\end{equation}
A Killing spinor which in addition satisfies the condition
$$
\nabla_{A'}{}^B K_{AB} + \nabla_A{}^{B'} \bar K_{A'B'} = 0
$$
corresponds via
$$
K_{ab} = i (K_{AB} \eps_{A'B'} - \bar K_{A'B'} \eps_{AB})
$$
to a Killing-Yano tensor, i.e. a skew 2-tensor	$K_{ab} =
K_{[ab]}$
satisfying the
Killing-Yano equation
$$
\nabla_{(a} K_{b) c} = 0 ,
$$
see
\cite{penrose:1973,carter:mclenaghan:1979,collinson:1976,jezierski:lukasik:2009}
for further information.

Let $\phi_{AB}$ be the Maxwell
spinor, i.e. a solution of the massless spin 1 equation. Then
$
\phi_{AB} K^{AB} = \Psi_2^{-1/3} \phi_1 ,
$
where $\phi_1$ is the spin weight $0$ Maxwell scalar and this rescaling of
$\phi_1$ solves a wave equation with potential
\begin{equation}\label{eq:fackerell-ipser-intro}
(\square + 2 \Psi_2) (\Psi_2^{-1/3} \phi_1 )  = 0,
\end{equation}
where $\square = \nabla^a \nabla_a$ is the d'Alembertian of the spacetime and
$\Psi_2$ is the spin weight $0$ Weyl scalar of the background.
This is precisely
the wave equation derived for the rescaled spin-weight $0$ Maxwell scalar on
Kerr
by Fackerell and Ipser \cite{fackerell:ipser}, who also argued that
$\phi_1$ can be used as a potential for the Maxwell field on the Kerr
spacetime.

The special case of this spin weight $0$ wave equation for the Maxwell field
on the Schwarzschild spacetime was used
recently in the proof by Blue \cite{blue:maxwell} of decay estimates for the
Maxwell equation. In this work, Blue was
inspired by Price \cite{price:1972:II} who showed that the Regge-Wheeler
\cite{regge:wheeler:1957} wave
equation for axial perturbations of the Schwarzschild spacetime can be viewed
as a wave equation, with potential, for a rescaled version of the
imaginary part of the spin weight 0 linearized Weyl scalar.

\subsection{Gauge invariant equations for linearized gravity}
We follow the convention of Price and others in discussing perturbations of
GHP quantities and let
a subindex $A$ denote quantities defined on the background and indicate first
order perturbed quantities with a subindex $B$. See
\cite{stewart:walker:1974,breuer:book} for detailed treatments of
perturbation theory in the context of tetrad based formalisms.
The Regge-Wheeler wave
equation in the form derived by Price can be written as
$$
( \square + 8 \Psi_{2A}	 ) (\Psi_{2A}^{-2/3} \Im \Psi_{2B} ) = 0 ,
$$
on the Schwarzschild spacetime, where in Schwarzschild coordinates $\Psi_{2A}
= -M r^{-3}$.

The approach of Price does not generalize directly to include
polar
perturbations.	These had previously been treated by Zerilli
\cite{zerilli:1970}, see also Moncrief \cite{moncrief:1974},
who derived a non-local
wave equation governing a gauge invariant potential for the polar degrees of
freedom,
and related work by Bi{\v{c}}{\'a}k \cite{bicak:1979} on
perturbations of the Reissner-Nordstr\"om solution.
Among the difficulties in generalizing Price's approach to the Regge-Wheeler
equation to cover general perturbations even of Schwarzschild is that in the
background the $\Psi_{2A}$ is non-zero and hence $\Psi_{2B}$ fails to be
gauge invariant.
In fact, it is only the linearized Weyl scalars
$\Psi_{0B}, \Psi_{4B}$
of extreme spin weights $2,-2$,
which are both coordinate and tetrad-gauge invariant, and
hence it is only those which can be directly viewed as physically measurable
quantities.

Teukolsky \cite{teukolsky:1972,teukolsky:1973:I} showed, by calculating in the
NP formalism on a Kerr background, working in the principal Kinnersley tetrad,
that suitably
rescaled versions of
the extreme spin weight
Maxwell ($s=1,-1$) and linearized Weyl scalars ($s=2,-2$)
on the Kerr spacetime satisfy decoupled
and separable
wave equations. The resulting system
is usually called the Teukolsky Master Equation (TME).

Shortly after the work of Teukolsky, Ryan \cite{ryan:1974}
showed that the vacuum Teukolsky system for linearized gravity
can be derived simply by projecting the Penrose
wave equation, i.e. the covariant tensor wave equation
$$
\square R_{abcd} = R_{abef}R_{cd}{}^{ef} +2(R_{aecf}R_b{}^e{}_d{}^f -R_{aedf}
R_b{}^e{}_c{}^f) ,
$$
satisfied by the
curvature tensor,
on a principal null tetrad, and linearizing.
This theme has been taken up and generalized
to arbitrary vacuum backgrounds by Bini et al. \cite{bini:etal:2002,bini:etal:2003}, where all equations for the $\Psi$'s and $\phi$'s are
calculated from components of a generalized de Rham operator acting on the
Riemann and Maxwell tensor.

Now consider a general vacuum type D background
spacetime.
Working in a principal tetrad, let $\psi_s$
be one of the fields $\Psi_{0B}, \Psi_{2A}^{-4/3} \Psi_{4B},
\phi_0, \Psi_{2A}^{-2/3} \phi_2$ (letting $s$ be the spin weight of the
field).
Further, let
$$
B_a = - (\NPrho \NPn_a - \NPtau \NPmbar_a) .
$$
Define a generalized wave operator acting on properly weighted quantities
by
$$
\squareTME_p = g^{ab} (\Theta_a + p B_a ) (\Theta_b + p B_b) ,
$$
where $\Theta_a$ is the weighted GHP covariant derivative, see section
\ref{sec:preliminaries} for details.
Then, cf. \cite{bini:etal:2002},
the TME$_s$ for spin s fields can be written in the form
\begin{equation}\label{eq:TME-intro}
( \squareTME_{2s} - 4s^2 \Psi_{2A} ) \psi_s = 0 .
\end{equation}

As shown by Teukolsky, for the Kerr case, equation \eqref{eq:TME-intro} can be
separated into
radial and angular equations.
Due to this fact, it has been possible to formally analyze
the Teukolsky system and its solutions and many
interesting discoveries have been made. Among these are the relation of the
separability of the system to the presence of the Carter constant and
corresponding symmetry operators, the calculation of the separation
constants,
as well as the proof of mode stability for
the Teukolsky system \cite{whiting:1989}.

However, it is not clear from these works that the Teukolsky system is well
suited for the analysis of the asymptotic decay properties of the higher spin
fields. Among the difficulties encountered in attempting to analyze the
Teukolsky system are the facts that it
has a long-range potential and
lower-order terms with slowly decaying coefficients.
See \cite{Hughes:2000pf} and citations therein for discussion.

\subsection{Spin lowering the linearized Weyl field}
Recall that
on a vacuum
type D spacetime,
the spin-lowered Weyl field $\psi_{ABCD} K^{CD}$, where $K_{AB}$ is a
Killing spinor, satisfies the Maxwell
equation, see \cite[\S 3.8]{jordan:ehlers:sachs:1961:II}, see also
\cite{penrose:1975}.
The same statement holds for the linearized Weyl
spinor\footnote{Here we use a $\delta$ to denote linearized quantities,
  eg. $(\delta \psi)_{ABCD}$ in order to avoid confusion with spinor indices.}
$(\delta \psi)_{ABCD}$
on Minkowski space,
and further in that case lowering the spin by 2 gives
$(\delta \psi)_{ABCD} K^{AB} K^{CD}$ which is a solution to the spin 0 wave
equation. We now consider the equations satisfied by these fields on a vacuum
type D background.

Let $(\delta \psi)_{ABCD}$
be the linearized Weyl spinor on a vacuum  type D
background,
and let $K_{AB}$ be the Killing spinor as in
\eqref{eq:killing-spinor-intro}. Then, expanding the spin 1 field
$ (\delta \psi)_{ABCD} K^{CD} $
which corresponds to a skew 2-tensor, into weighted scalars gives 
the rescaled linearized Weyl scalars $\hat \phi_{i-1} = \Psi_{2A}^{-1/3} \Psi_{iB}$,
$i=1,2,3$, of spin weights $1,0,-1$ and additional terms arising from linearized tetrad. 
Thus, by analogy with the above,
it is reasonable to suppose that the fields
$\hat\phi_0 = \Psi_{2A}^{-1/3}\Psi_{1B}$, $\Psi_{2A}^{-2/3} \hat\phi_2 =
\Psi_{2A}^{-1} \Psi_{3B}$ of spin weights $1, -1$, respectively,
satisfy an analogue of the Teukolsky system
TME$_s$, for $s=1, -1$, while the spin weight 0 field
$\hat\phi_1 = \Psi_{2A}^{-2/3} \Psi_{2B}$ can be expected to
satisfy an analogue of the Fackerell-Ipser equation.
%%%
Further,
since the non-extreme linearized Weyl scalars fail
to be gauge invariant, one also expects the corresponding equations to
contain gauge potential terms.

For the spin weight one case, a calculation, see
section \ref{sec:pert-calc},
shows
$$
(\squareTME_2 -4\Psi_2 \big)( \Psi_{2A}^{-1/3} \Psi_{1B}) =
-6\Psi_{2A}^{2/3}[ ({\tho}'+2\rho'-\bar{\rho}')\kappa_B
  -({\edt}'+2\tau'-\bar{\tau})\sigma_B + 2\Psi_{1B}],
$$
The equation satisfied by $\Psi_{2A}^{-1} \Psi_{3B}$ is similar.
The right hand side of the
equation for $\Psi_{1B}$
corresponds
to the right hand
sides of the equations for $\phi_{0B}$ on a charged type D
background with $\Psi_{2A}$ playing the role of
the spin weight 0 Maxwell scalar $\phi_{1A}$, see section
\ref{sec:pert-calc}.
The analogous statement holds for $\Psi_{3B}$ after applying a prime.
As we shall see, the conditions that the right hand sides of the equations for
$\Psi_{1B}, \Psi_{3B}$ are zero are tetrad gauge conditions. For the
case of the perturbed
Maxwell field this corresponds
to turning off the background charge, i.e. imposing
the condition $\phi_{1A} = 0$.
The fact that the just mentioned tetrad gauge conditions can be viewed as the
``ghost'' of the background charge motivated Chandrasekhar
\cite[p.240]{chandra:book} to use the
term \emph{phantom gauge}.
In fact, Chandrasekhar
showed that by imposing the phantom
gauge, the rescalings $\Psi_{2A}^{-1/3}
\Psi_{1B}, \Psi_{2A}^{-1} \Psi_{3B}$ satisfy the TME$_s$, for $s=1,-1$,
respectively.

We show in section \ref{sec:gauge} below that
the phantom gauge can be viewed as prescribing
a gauge source function for the tetrad degrees of freedom along the lines of
Friedrich \cite{friedrich:1985},
with the linearized Weyl field
itself as part of the gauge source. The phantom gauge was studied by
Chandrasekhar
from a formal point of view only, and the possible implications of
this procedure for the hyperbolicity and well-posedness of the linearized
Einstein equations were not analyzed in his work. We show here
that the phantom gauge condition is compatible with a
well posed system of equations for linearized gravity.

Finally, we consider the spin weight $0$ linearized Weyl scalar.
As the background $\Psi_{2A}$ is
non-vanishing, one has that $\Psi_{2B}$ is coordinate gauge dependent (but
tetrad-gauge independent).
Motivated by the previous discussion we consider the equation satisfied by the
spin weight 0 field $\Psi_{2A}^{-2/3} \Psi_{2B}$ obtained by lowering the
spin of the
the linearized Weyl field by 2.
A calculation, cf. section \ref{sec:pert-calc},
shows that this rescaled spin weight 0 linearized Weyl scalar solves
the equation
\begin{equation}\label{eq:spin-0-intro}
( \square +8\Psi_2 ) (\Psi_{2A}^{-2/3}\Psi_{2B} ) = - 3\square_B  \Psi_{2A}^{1/3} ,
\end{equation}
where the right hand side is the first order perturbation of the wave
operator, acting on the \emph{background} spin weight zero Weyl scalar.
In section \ref{sec:ginvschw}
we show that, restricting to the Schwarzschild case,
equation \eqref{eq:spin-0-intro}
contains all of the information in the Regge-Wheeler and Zerilli-Moncrief
systems by giving a direct derivation of these systems starting from
\eqref{eq:spin-0-intro}.

The condition that the right hand side of \eqref{eq:spin-0-intro}
vanishes can be viewed as a generalized harmonic
coordinate condition.
It is worth noting that this gauge condition can be
imposed also in the Schwarzschild case. Imposing this generalized harmonic gauge
condition, the spin weight zero linearized Weyl scalar
satisfies the scalar wave equation
\begin{equation}\label{eq:gen-RW-intro}
( \square +8\Psi_{2A} ) (\Psi_{2A}^{-2/3}\Psi_{2B} ) = 0 ,
\end{equation}
which can be viewed as a generalization of Price's version of the
Regge-Wheeler equation not
only to the full set of perturbations of Schwarzschild but also to
perturbations of Kerr.
Lun and Fackerell
\cite{lun:fackerell:1975}
considered the situation on Schwarzschild and argued formally
that by imposing a suitable gauge condition, one obtains equation
\eqref{eq:gen-RW-intro} (specialized to the Schwarzschild case).

We further point out that a generalized harmonic gauge condition with a gauge
source function involving $\Psi_{2B}$ can also be
used to modify the potential in \eqref{eq:gen-RW-intro} so that the equation
becomes the Fackerell-Ipser equation. However, as already discussed by
Crossman and Fackerell \cite{crossman:fackerell:1982,fackerell:1982}
this is possible only in the
\emph{rotating case}, and in particular for the Kerr family of spacetimes
involves a division by $a$. As in the case of \cite{lun:fackerell:1975}, the
discussion in  the papers \cite{crossman:fackerell:1982,fackerell:1982} is
quite formal and the gauge conditions
are there not expressed in terms of
gauge source functions. It is interesting to note that the just mentioned
work of Crossman
and Fackerell took as a starting point the Maxwell equation for the
spin-lowered Weyl field in a type D spacetime, and its linearization. This
has been carried further in the work of Ferrando et
al. \cite{fayos:ferrando:jaen:1990} where gauge
conditions yielding an \emph{exact} Maxwell system for the linearized,
spin-lowered Weyl field have been considered.

Now, following the approach taken by Blue for the case of Maxwell on
Schwarzschild, where the spin weight zero scalar was used as a potential for
the Maxwell field, it is an interesting possibility to use $\Psi_{2A}^{-2/3}
\Psi_{2B}$ as a potential for the full linearized gravity system. This would
allow one to reduce decay estimates for linearized gravity on Kerr to the
study of the scalar wave equation with potential \eqref{eq:gen-RW-intro}.

As is indicated by the discussion above, there is a
great deal of freedom in using gauge conditions to change the nature of the
(tetrad based) linearized Einstein equations. It is to be expected that this remark applies
equally to the full, non-linear system of Einstein equations. The
implications of this remain to be considered.
One could in principle go further, and make use of the gauge dependence of
$\Psi_{2B}$ to remove the potential from the equation and achieve a setup (in
the rotating case)
where $\Psi_{2A}^{-2/3} \Psi_{2B}$ solves the scalar wave equation
$
\square (\Psi_{2A}^{-2/3} \Psi_{2B}) = 0 .
$
We point out that the gauge conditions chosen by Chandrasekhar, cf.  \cite[\S
  82]{chandra:book}, in his
considerations of linearized gravity on the Kerr background included the
conditions $\Psi_{1B} = \Psi_{2B} = \Psi_{3B} = 0$. This type of gauge
condition will not be considered in detail here.

For linearized gravity on Kerr, among the questions which should be
considered are the choice of potential for linearized gravity and the field
equation governing this. For the gauge invariant scalars
$\Psi_{0B}, \Psi_{4B}$ these issues have been extensively discussed in the
literature. Making use of gauge conditions as discussed above opens up
interesting new possibilities. However to make full use of these, the problem
of reconstructing the full solution of linearized gravity from e.g.
$\Psi_{2B}$ must be considered.

\subsection{Overview of this paper} The plan of this paper is as follows. In
section \ref{sec:preliminaries} we set up notation and give a brief overview
of the GHP formalism and its specialization to vacuum
type D backgrounds. We also
discuss there gauge issues that arise when working in a tetrad based
formalism. Section \ref{sec:eq-lin-grav} starts by introducing the properly
weighted generalized wave operators which occur in the GHP formalism and in
the Teukolsky system. Further, we give there a derivation in the GHP
formalism of the equations for linearized gravity, cf. section
\ref{sec:pert-calc}.
The gauge nature of the
non-trivial right hand sides of the equations for the scalars of
non-extreme spin weights is discussed in section \ref{sec:gauge} where we also
give gauge-fixed versions of these systems which lead to new potentials for
linearized gravity satisfying well-posed field
equations.
In section \ref{sec:ginvschw} the gauge invariant Regge-Wheeler and Zerilli equations for Schwarzschild background are derived from \eqref{eq:spin-0-intro}.

\section{Preliminaries and notation} \label{sec:preliminaries}
We use the conventions and notations of \cite{GHP}. In particular we use
abstract index notation with lower case latin indices for tensors and upper
case latin indices for spinors. For tetrad indices we
use lower case fraktur font, while for
coordinate indices we use lower case greek letters.
Unless otherwise stated we shall consider only vacuum spacetimes $(M,
g_{ab})$ of dimension $4$, with signature $+---$.

\subsection{GHP formalism} \label{sec:GHP}
The Geroch-Held-Penrose (GHP) null-tetrad formalism \cite{GHP}
allows one to represent
the Einstein equations in a compact form, and gives an efficient tool for
calculations. Since we will make heavy use of this formalism and its
properties, we give, in order to make the paper reasonably self-contained,
a brief description of its main features.

Consider a null tetrad
$(e_\fraka^a) =
(\NPl^a, \NPn^a, \NPm^a, \NPmbar^a)$
consisting of two
real null vectors
$\NPl^a, \NPn^a$
and two complex linear combinations
of spatial vectors
$\NPm^a, \NPmbar^a$,
normalized such that
the only non-vanishing inner products of the tetrad vectors are
$$
\NPl^a \NPn_a = - \NPm^a \NPmbar_a = 1 .
$$
In order to avoid clutter we suppress the abstract index on the tetrad
vectors $e_\fraka$ when convenient.
The coframe $e^\fraka$ is defined by the relations $e^\fraka(e_\frakb) =
\delta^\fraka{}_\frakb$.
We note the useful relations
\begin{subequations}\label{eq:gab-tetr}
\begin{align}
g_{ab} &= \NPl_{a} \NPn_{b} + \NPn_a \NPl_b - \NPm_a \NPmbar_b - \NPmbar_a
\NPm_b ,
\\
\delta^a{}_b &= \NPl^a \NPn_b + \NPn^a \NPl_b - \NPm^a \NPmbar_b - \NPmbar^a
\NPm_b . \label{eq:kronecker}
\end{align}
\end{subequations}
A choice of a null tetrad picks out a 2-dimensional subgroup of
the Lorentz group in each tangent space (and hence also
a reduction of the principal
$\SO_+(3,1)$ bundle of $(M,g)$, see below)
which preserves the null planes spanned by
$\NPl^a, \NPn^a$
and the spatial planes spanned by
$\NPm^a, \NPmbar^a$.
These can be represented in terms of a non-vanishing complex field
$\lambda$ by the \emph{boost rotations}
\begin{subequations}\label{eq:tetrad-weights}
\begin{align}
\NPl^a &\to \lambda\bar\lambda \NPl^a, \quad \NPn^a \to
\lambda^{-1} \bar \lambda^{-1}	\NPn^a , \\
\intertext{and the \emph{spin rotations}}
\NPm^a &\to \lambda \bar\lambda^{-1} \NPm^a, \quad
\NPmbar^a \to \lambda^{-1} \bar \lambda \NPmbar^a .
\end{align}
\end{subequations}
Projecting tensor fields on the spacetime on the null tetrad gives a
representation of these fields in terms of collections of tetrad
components which are simply complex fields on the spacetime.
In general, a scalar
field $\eta$
defined by projecting a
tensor field will transform as
$$
\eta \to \lambda^p \bar\lambda^q \eta \,,
$$
for some integers $p,q$,
under the above defined action of $\lambda$. A quantity $\eta$
which transforms according to the above rule is said to have \emph{type}
$\{p,q\}$ and fields with well defined type are referred to as
\emph{weighted quantities}.
Note the notion of weighted quantity extends to tensors.
In particular, the
tetrad vectors
$\NPl^a, \NPn^a, \NPm^a, \NPmbar^a$
have types $\{1,1\}$,
$\{-1,-1\}$, $\{1, -1\}$, $\{-1, 1\}$, respectively.
It is useful to note that the type is additive under multiplication
and hence the weighted quantities form a
graded algebra.
The spin and boost weights $s,r$ of a weighted
quantity $\eta$ are
$s = \half(p-q)$, $r = \half(p+q)$.

The following formal operations take weighted quantities to weighted
quantities,
\begin{equation}\label{eq:ghpsym}
\begin{aligned}
^-(\text{bar})&: \; \NPl^a \to \NPl^a, \;   \NPn^a \to \NPn^a, \;  \NPm^a
\to \NPmbar^a,	\; \NPmbar^a \to \NPm^a, &\{p,q\}\to\{q,p\} ,\\
'(\text{prime})&: \; \NPl^a \to \NPn^a, \;   \NPn^a \to \NPl^a, \;
\NPm^a \to \NPmbar^a,  \; \NPmbar^a \to \NPm^a, &\{p,q\}\to\{-p,-q\} ,\\
^*(\text{star})&: \; \NPl^a \to \NPm^a , \;  \NPn^a \to -\NPmbar^a, \;
\NPm^a \to -\NPl^a, \; \NPmbar^a \to \NPn^a, &\{p,q\}\to\{p,-q\} .
\end{aligned}
\end{equation}
The bar \ $\bar{}$ \ operation acting on a weighted quantity is simply the
complex conjugation of the field.
We  have
$\bar{\bar{\eta}} = \eta$ and $\eta'' = (-1)^{p+q} \eta$ (note however that
we shall consider only fields with $p+q$ even). Further,
the bar $\bar{}$  and prime ${}^\prime$ operations commute,
while the star
${}^*$ operation commutes with neither of these. Thus, the star operation has
to be treated separately from the bar and prime operations.

For the case of an orthonormalized tetrad, the Levi-Civita connection can be
represented in terms of 24 independent real connection coefficients. In terms
of the null tetrad introduced above, the connection coefficients
$\Gamma_{\fraka \frakb}^\frakc = e^\frakc_b (
\nabla_{e_\fraka} e_\frakb^b )$ combine into 12 complex scalars,
called
\emph{spin coefficients}. Of these only 8
are properly weighted, and can be represented in terms of the quantities
\begin{align}
\kappa &=  \NPm^b \NPl^a \nabla_a \NPl_b , \quad
\sigma =  \NPm^b \NPm^a \nabla_a \NPl_b, \quad
\rho   =  \NPm^b \NPmbar^a \nabla_a \NPl_b , \quad
\tau   =
\NPm^b \NPn^a \nabla_a \NPl_b ,
\label{eq:spincoeff-def}
\end{align}
and their primes $\kappa', \sigma', \rho', \tau'$. The effect of the star
operation on the spin coefficients and their primes and complex conjugates
can be calculated directly from
\eqref{eq:spincoeff-def}, or see \cite[p. 878]{GHP} for a list.
The types of the spin
coefficients are
$$
\kappa: \{3,1\}, \quad \sigma: \{3,-1\}, \quad \rho: \{1,1\}, \quad \tau: \{1,-1\} ,
$$
and the types of their primes are given according to \eqref{eq:ghpsym}.

The remaining 4 spin coefficients
\begin{align*}
\beta &= \half ( \NPn^b \NPm^a \nabla_a \NPl_b - \NPmbar^b \NPm^a \nabla_a
\NPm_b ) ,\\
\eps &=	 \half ( \NPn^b \NPl^a \nabla_a \NPl_b - \NPmbar^b \NPl^a \nabla_a
\NPm_b ) ,
\end{align*}
and their primes $\beta', \eps'$,
are not
properly weighted and are not used explicitly in the GHP formalism.

Let $\Co_*$ denote the non-zero complex numbers.
As discussed by Ehlers \cite{ehlers:1974}, see also
\cite{guven:1976, harnett:1990}, the choice of two null directions gives a
reduction of the principal $\SO_+(3,1)$ frame bundle of the spacetime to a
principal $\Co_*$ bundle $B$ with the action of $\Co_*$ $z : (\lambda,
e_\fraka)\mapsto e_\fraka . z(\lambda)$ given by
\eqref{eq:tetrad-weights}. The weighted quantities may be viewed as sections
of associated complex line bundles\footnote{By the same construction we can
also treat e.g. differential forms as weighted quantities.}
$\EE^{\{p,q\}} = B \times_{z^{\{p,q\}}} \Co$
determined by the representation $z^{\{p,q\}}:
(\lambda, v) \to  \lambda^p \bar{\lambda}^q v
$ of $\Co_*$ on $\Co$. The restriction of the Levi-Civita connection to the
reduced bundle $B$ induces
a connection on the
bundles $\EE^{\{p,q\}}$ given by
\begin{equation}\label{eq:Theta-form}
\Theta_a \eta = \nabla_a \eta - p \omega_a \eta - q \bar{\omega}_a \eta ,
\end{equation}
where
$\omega_a$ is the connection form
\begin{align*}
\omega_a &= - \eps' \NPl_a + \eps \NPn_a + \beta' \NPm_a - \beta \NPmbar_a
  \\
&= \half ( \NPn^b \nabla_a \NPl_b + \NPm^b \nabla_a \NPmbar_b ) .
\end{align*}
Under a gauge transformation  $\omega_a$ transforms as
$$
\omega_a . z = \omega_a + \frac{\nabla_a\lambda}{\lambda} ,
$$
Note that $\omega_a' = -\omega_a$. It follows that
$$
(\Theta_a \eta)' = \Theta_a  \eta' ,
$$
for any properly weighted quantity $\eta$.

The GHP operators $\tho, \tho', \edt, {\edt}'$ are defined as
the $\Theta_a$ covariant derivative along the tetrad vectors,
$$
\tho = \NPl^a \Theta_a, \quad {\tho}'= \NPn^a \Theta_a, \quad
\edt = \NPm^a \Theta_a, \quad {\edt}'=\NPmbar^a \Theta_a .
$$
The action of the bar, prime and star operations on the GHP operators follows
from their action on the tetrad vectors.
Expanding $\Theta_a$ in terms of the GHP operators gives
\begin{equation}\label{eq:Theta-GHPop}
\Theta_a =
\NPl_a {\tho}' + \NPn_a \tho -\NPm_a {\edt}' -\NPmbar_a \edt.
\end{equation}
In terms of the graded algebra of weighted quantities, the covariant
derivative $\Theta_a$ as well as the GHP operators satisfy a (graded) Leibniz
rule. Further,
as remarked above, the notions of properly weighted quantity extend to
differential forms and more general objects, and hence also the weighted
covariant derivative lifts to act on such objects. In particular, the tetrad
elements themselves are properly weighted quantities,
and the action of $\Theta_a$ and
the GHP operators on these can be read off from the definitions. We have
for example $\tho \NPl_c = \delta_c{}^b \NPl^a \Theta_a \NPl_b$. Expanding
this out using \eqref{eq:kronecker}, \eqref{eq:Theta-form} and the
definitions of the spin-coefficients gives $\tho \NPl_c = - \bar{\kappa}
\NPm_c - \kappa \NPmbar_c$. The following equations
\begin{subequations}\label{eq:dtetrad}
\begin{align}
\tho l_a &= -\bar{\kappa} m_a -\kappa \bar{m}_a,
& \tho m_a = -\bar{\tau}'l_a -\kappa n_a ,\\
{\tho}' l_a &= -\bar{\tau} m_a -\tau \bar{m}_a, & {\tho}' m_a = -\bar{\kappa}'l_a -\tau n_a ,\\
\edt l_a &= -\bar{\rho} m_a -\sigma \bar{m}_a,
& \edt m_a = -\bar{\sigma}'l_a -\sigma n_a ,\\
{\edt}' l_a &= -\bar{\sigma} m_a -\rho \bar{m}_a,
& {\edt}' m_a = -\bar{\rho}'l_a -\rho n_a .
\end{align}
\end{subequations}
and their primes
and complex conjugates
give the complete set of relations.

The Leibnitz rule for covariant derivative together with
\eqref{eq:dtetrad} allows one to read off the action of the GHP operators on
tensors projected on any combination of tetrad vectors, and hence covariant
tensor equations can be expressed equivalently
as collections of scalar equations
in the GHP formalism.

The 10 degrees of freedom of the Weyl tensor can be represented by the five
weighted Weyl scalars
\begin{align*}
\Psi_0 &= W_{abcd} \NPm^a \NPl^b \NPl^c \NPm^d ,
\quad
\Psi_1 = W_{abcd} \NPn^a \NPl^b \NPl^c \NPm^d,
\quad
\Psi_2 = W_{abcd} \NPm^a \NPl^b \NPmbar^c \NPn^d,
\\
\Psi_3 &= W_{abcd} \NPn^a \NPl^b \NPmbar^c \NPn^d,
\quad
\Psi_4 = W_{abcd} \NPmbar^a \NPn^b \NPn^c \NPmbar^d
\,.
\end{align*}
Similarly, the Maxwell field strength can be represented by the 3 Maxwell
scalars
$$
\phi_0= F_{ab} \NPl^a \NPm^b ,
\quad \phi_1= \half (F_{ab} \NPl^a \NPn^b + F_{ab} \NPmbar^a \NPm^b ) ,
\quad \phi_2= F_{ab} \NPmbar^a \NPn^b \, .
$$
We shall refer to the spin coefficients and the Weyl and Maxwell scalars
collectively as GHP quantities.

The Weyl scalars $\Psi_i$, $i=0,\dots,4$ have types $\{4-2i,0\}$ while
the Maxwell scalars $\phi_i$, $i=0,1,2$ have types $\{2-2i,0\}$. The prime
operation gives $\Psi'_i =  \Psi_{4-i}$, $i = 0,\dots,4$, and $\phi'_i = -
\phi_{2-i}$, $i=0,1,2$. Further, $\Psi_i^*=\Psi_i$ and $\phi_i^*=\phi_i$.

We now state, modulo prime and star operations,
the Einstein, Bianchi and Maxwell equations in GHP notation, specialized to
the vacuum case. Working in a tetrad formalism, the Einstein equation takes
the form of a system of first-order equations for the connection
coefficients, i.e. in the GHP setting, for the spin coefficients. These are
given by
\begin{subequations}\begin{align}
\edt \rho -{\edt}' \sigma &= (\rho -\bar{\rho})\tau +(\bar{\rho}'-\rho')\kappa -\Psi_1 ,\label{eq:ghpeinsteina}\\
\tho \rho -{\edt}' \kappa &= \rho^2 + \sigma \bar{\sigma} -\bar{\kappa}\tau
-\kappa \tau' , \label{eq:ghpeinsteinb}\\
\tho \sigma -\edt \kappa &= (\rho +\bar{\rho})\sigma -(\tau
+\bar{\tau}')\kappa +\Psi_0 , \label{eq:ghpeinsteinc}\\
\tho \rho' -\edt \tau' &= \rho' \bar{\rho} +\sigma \sigma' -\tau'\bar{\tau}'
-\kappa \kappa' -\Psi_2 , \label{eq:ghpeinsteind}
\end{align}\end{subequations}
together with their primed and starred versions.
The Bianchi equations are given by
\begin{subequations}\label{eq:ghpbianchi}
\begin{align}
(\tho -4 \rho) \Psi_1 -({\edt}'-\tau')\Psi_0 &= -3\kappa \Psi_2 , \label{eq:ghpbianchia}\\
(\tho -3 \rho) \Psi_2 -({\edt}'-2\tau')\Psi_1 &= \sigma' \Psi_0 -2 \kappa
  \Psi_3 ,  \label{eq:ghpbianchib}
\end{align}\end{subequations}
together with their primed and starred versions,
and the Maxwell equations are
\begin{align}
(\tho -2\rho)\phi_1 -({\edt}'-\tau')\phi_0 = -\kappa \phi_2 , \label{eq:ghpmaxwell}
\end{align}
with its primed and starred versions.
Further, the GHP operators acting on weighted quantities
satisfy the commutator relations
\begin{subequations}
\begin{align}
\big[\tho,{\tho}'\big]\eta &= \big[ (\bar{\tau} -\tau')\edt +(\tau -\bar{\tau}'){\edt}' -p(\kappa \kappa' -\tau \tau' +\Psi_2) -q(\bar{\kappa} \bar{\kappa}'-\bar{\tau}\bar{\tau}' +\bar{\Psi}_2) \big] \eta , \label{eq:ghpcomma} \\
\big[\tho,\edt\big] \eta &= \big[-\bar{\tau}'\tho -\kappa {\tho}' +\bar{\rho}\edt +\sigma {\edt}' -p(\rho' \kappa -\tau' \sigma +\Psi_1) -q(\bar{\sigma}'\bar{\kappa} -\bar{\rho}\bar{\tau}') \big] \eta , \label{eq:ghpcommb}
\end{align}
\end{subequations}
together with their primed and starred versions.

\subsection{Petrov type D spacetimes}
\label{sec:typeD}
In a spacetime of type D we can fix a null tetrad up to rescalings (and a
trivial rearrangement) by aligning the real null vectors with the principal
null directions. Such a tetrad is called a principal tetrad.
In a vacuum type D spacetime, working in a principal tetrad,
$\Psi_0=\Psi_1=\Psi_3=\Psi_4=0$ follows and due to the Goldberg-Sachs theorem $ \kappa=\kappa'=\sigma=\sigma'=0$. The only non-vanishing GHP quantities are
$$
\Psi_2, \rho, \tau, \rho', \tau',
$$
and
the Bianchi identities \eqref{eq:ghpbianchi} simplify to
\begin{equation}\label{eq:ghpDbianchi}
{\tho}\Psi_2=3\rho\Psi_2 , \quad  {\edt}\Psi_2=3\tau\Psi_2 ,
\end{equation}
together with their primed versions.
See \cite{edgar:etal:2009} for further identities valid for the GHP quantities
valid in vacuum type D spacetimes.
We record for later use a
commutation relation for $\{p,0\}$ quantities $\eta$ on vacuum type D
backgrounds. The following identity (and its prime) is a consequence of
\eqref{eq:ghpcommb},\eqref{eq:ghpeinsteina}
and
\eqref{eq:ghpeinsteina}$^*$,
\begin{equation} \label{advcomm}
\big[\tho -a\rho , \edt-a\tau \big] \eta
= \bar{\rho}(\edt-a\tau)\eta -\bar{\tau}' (\tho-a\rho)\eta,
\end{equation}
see  \cite[equation (2.5)]{fernandes:lun:1997} for a more general relation
involving $q$-weight.

\subsection{Perturbation theory and gauge transformations}
\label{sec:pert-gauge}
Here we give a short overview of gauge transformations in perturbation
theory.	 See
\cite{stewart:walker:1974, breuer:book} for more details.
Perturbations of a spacetime can be understood in terms of curves
in the space of
solutions of Einstein field equations, originating at a given background
spacetime. Linear perturbations are
tangents to such curves at the origin. We denote the perturbation
parameter by $\epsilon$.

The identification of points of background and perturbed spacetime is called identification gauge. Introducing coordinates $x^a$ in the background, an infinitesimal transformation of the form $x^a \to x^a +\epsilon \xi^a$ can be interpreted as changing the
identification of points between background and perturbed
spacetime. Quantities which do not change under these transformations are
called identification or coordinate gauge invariant. A quantity is
\textit{coordinate gauge invariant} if and only if
it vanishes or is a constant scalar
or a constant linear combination of products of Kronecker deltas in the
background, see \cite[p.24]{stewart:book}.

The Weyl scalars $\Psi_{iB}$ transform as scalars under coordinate
transformations,
$$
\Psi_{iB} \to \Psi_{iB} - \eps\, \xi^\mu \partial_\mu\Psi_{iA} + O(\eps^2).
$$
In a type D background, only $\Psi_2$ is non-zero, and hence
the spin 0 scalar is the only one of the $\Psi_{iB}$, $i=0,\dots,4$ which
fails to be coordinate gauge invariant.

As mentioned in section \ref{sec:GHP},
we choose the background tetrad to be fixed up to a two-dimensional subgroup of
the Lorentz group corresponding to boost rotations
of
the future pointing null vectors $\NPl^a, \NPn^a$,
and spin rotations of $\NPm^a, \NPmbar^a$.
However, the
perturbed tetrad has the full
transformation freedom under infinitesimal elements of the Lorentz group
$$
\begin{array}{lll}
\NPl^a_B \to \NPl^a_B, & \NPn^a_B \to \NPn^a_B	+\eps (\bar{a}\NPm^a +a\NPmbar^a),
& \NPm^a_B \to \NPm^a_B +\epsilon a \NPl^a ,\\
\NPl^a_B \to \NPl^a_B +\eps(\bar{b}\NPm^a + b\NPmbar^a), & \NPn^a_B \to \NPn^a_B,
& \NPm^a_B \to \NPm^a_B +\eps b \NPn^a , \\
\NPl^a_B \to \NPl^a_B+\eps A \NPl^a, & \NPn^a_B \to \NPn^a_B -\eps A\NPn^a, &
\NPm^a_B \to \NPm^a_B-\i\eps \Theta\NPm^a , \\
\end{array}
$$
where $a,b$ are complex and $A,\Theta$ are real functions (see e.g. the linearized versions of \eqref{eq:tetrad-weights} for the third line.
A quantity which is invariant under these transformations will be called \textit{tetrad gauge invariant}.
For the first subset of infinitesimal Lorentz transformations we have for example
\begin{align}
\Psi_{jB} \to \Psi_{jB} +\epsilon j \bar{a} \Psi_{j-1\, A} , \quad j=0,...4; \Psi_{-1 \, A} =0 .
\end{align}
A complete table for all GHP quantities
can be found in \cite[\S
  5.10]{breuer:book}. It can be verified that $\Psi_{0B}, \Psi_{2B}$ and
$\Psi_{4B}$ are tetrad gauge invariant.

In the following, in order to avoid clutter in the notation, we will drop the index $A$ for background quantity, unless
it is not clear from the context whether a certain quantity is evaluated on
the background.

\section{Equations for linearized gravity and electromagnetism}
\label{sec:eq-lin-grav}
In this section, we derive equations for linear perturbations of the Weyl
components $\Psi_0, ..., \Psi_4$ on vacuum type D backgrounds, as well as for
the linearized Maxwell scalars $\phi_0, \phi_1, \phi_2$ on charged type D
backgrounds.
The gauge invariant fields
$\Psi_{0B}, \Psi_{4B}$ satisfy the Teukolsky system TME$_{s}$, for $s=2,-2$,
while
the equations for the tetrad gauge dependent scalars
$\Psi_{1B}, \Psi_{3B}$ correspond to the TME$_s$ for
$s=1,-1$, but with a non-trivial right hand side involving a gauge source
function, cf. section \ref{sec:gauge}.
For the spin weight zero linearized Weyl scalar $\Psi_{2B}$ we find a
new wave equation with a non-trivial right hand side. If the right hand side
vanishes, this equation is a direct generalization of the Regge-Wheeler wave
equation to type D. In section \ref{sec:gauge} below we consider the
structure of these right hand sides in more detail.

We point out that while the equations for the linearized
Weyl scalars as well as the
linearized Maxwell scalars decouple in the sense that each individual
equation involves only one of these scalars, the equations for the
non-extreme spins are coupled via linearized spin coefficients, unless
further gauge conditions are imposed.

\subsection{Weighted wave operators} \label{sec:weight-wave}
As is the case for any vector
bundle over $(M,g)$ with covariant derivative, there is a natural generalized
wave operator acting on sections of the bundles
$\EE^{\{p,q\}}$.
The Weyl and Maxwell scalars are properly weighted quantities of type $\{2s,
0\}$ for integer spin weights $s$. Since we shall be interested in
operators acting on the Weyl and Maxwell scalars, we restrict our attention to
the operator $\square_p = \Theta^a \Theta_a$ acting on quantities of type
$\{p,0\}$. Expanding this using \eqref{eq:Theta-GHPop} and \eqref{eq:dtetrad}
gives after some calculations using the commutation relations
\eqref{eq:ghpcomma},\eqref{eq:ghpcommb}
\begin{align}
\square_{p}
&= 2  \bigg{[} (\tho-\bar{\rho})({\tho}'-\rho')
  -(\edt-\bar{\tau}')({\edt}'-\tau') +\sigma\sigma' -\kappa'\kappa -\Psi_2
  \\
&\quad	+\frac{p}{2}(\kappa\kappa'-\tau\tau' +\rho\rho'-\sigma\sigma'+2\Psi_2)
   \bigg{]} , \label{box4}
\end{align}
Let $\omega_a$ be the connection form in $\Theta_a$, cf. section
\ref{sec:GHP}, and let $B_a$ be a properly weighted form of type
$\{0,0\}$. Then
$\omega_a - B_a$ is again a connection form on the weigted
bundles $\EE^{\{p,q\}}$ and $\Theta_a + pB_a + q \bar{B}_a$ is again a
weighted covariant derivative on the bundles $\EE^{\{p,q\}}$.
In particular, let
\begin{align}
B_a = - (\rho \NPn_a -\tau \NPmbar_a ) .\label{tmeconnection}
\end{align}
Modifying the covariant derivative with $B_a$ gives the weighted wave
operator on type $\{p,0\}$ quantities
$$
\squareTME_p = (\Theta^a + p B^a) (\Theta_a + p B_a) .
$$
Note that $\squareTME_0 = \square$.
We can now write the vacuum
Teukolsky master equation for a spin weight $s$ field
$\psi^{(s)}$
in the form
$$
[ \squareTME_{2s} - 4 s^2 \Psi_2 ] \psi^{(s)} = 0 ,
$$
cf. Bini et al. \cite[\S 4]{bini:etal:2003}.

A calculation shows that acting on a quantity of type $\{p,0\}$ we have
\begin{align}
\squareTME_{p}
&=\square_p + 2pB^a \Theta_a + p (\Theta^a B_a) + p^2 B^a B_a \nonumber
\\
&= 2(\tho -p\rho -\bar{\rho})({\tho}'-\rho')- 2(\edt-p\tau
-\bar{\tau}')({\edt}'-\tau') \nonumber \\
&\quad+ (p-2) [ \kappa \kappa' - \sigma \sigma'] +
(3p-2)\Psi_2 . \label{eq:squaretme-full}
\end{align}
Restricting to a type D background we have
\begin{multline}\label{eq:squaretme-restrict}
\squareTME_{p} =
2(\tho -p\rho -\bar{\rho})({\tho}'-\rho')- 2(\edt-p\tau
-\bar{\tau}')({\edt}'-\tau')  + (3p-2)\Psi_2 .
\end{multline}
since $\kappa, \sigma$ vanish there.

Recalling the discussion in section \ref{sec:GHP}, $\Theta_a$ transforms
properly under the prime operation. In particular, we have
$(\Theta_a \eta)' = \Theta_a \eta'$ and hence also
$(\Theta^a \Theta_a  \eta )' = \Theta^a \Theta_a \eta' $
for any properly weighted quantity. However, modified connection
$\Theta_a + p B_a$ does not satisfy this rule since
$B_a' \ne - B_a $.
Instead, the operator $\squareTME_p$ has the following transformation rule
involving rescalings.
\begin{lemma} \label{lem:tmeprime}
Let $\eta$ be a properly weigthed quantity of type $\{p,0\}$.
The generalized wave operator $\squareTME_{p}$ on a vacuum
type D background transforms under prime as
\begin{equation}\label{eq:tmeprime}
(\squareTME_{p} \eta)'=
\Psi_2^{-p'/3} \squareTME_{p'} (\Psi_2^{p'/3} \eta').
\end{equation}
\end{lemma}
\begin{remark}
It should be noted that the $p$ in $\squareTME_{p}$ denotes the weight of
the quantity on which it acts. Hence the $p'$ in the right hand side of
\eqref{eq:tmeprime} is $p' = -p$, since the type of $\eta'$ is $\{-p,0\}$.
\end{remark}
\begin{proof}
Using the commutation relations and field equations on type D, one gets the identity
\begin{equation}\begin{aligned}
&({\tho}'+p\rho'-\bar{\rho}')(\tho-\rho) -({\edt}'+p\tau'-\bar{\tau})(\edt-\tau) \\
& \hspace{2cm}= (\tho-\bar{\rho})({\tho}'+(p-1)\rho')-(\edt-\bar{\tau}')({\edt}'+(p-1)\tau') + 3p\Psi_2 .
\end{aligned}\end{equation}
Rescaling the RHS by $\Psi_2^{p/3}$, using Bianchi identities~\eqref{eq:ghpDbianchi} gives the result.
\end{proof}
As we shall see, using this transformation property for $\squareTME_p$, half
of the equations for linearized gravity discussed below follow without
calculation.

\subsection{Perturbation calculations} \label{sec:pert-calc}
We now derive the equations for the linearized Weyl scalars in terms of the
weighted wave operators.
\begin{thm} \label{thm:eq-lin-grav}
On a vacuum type D background we have
\begin{align}
\big[\squareTME_4 -16\Psi_2 \big] \Psi_{0B} &= 0 \label{psi0b},\\
\big[\squareTME_2 -4\Psi_2 \big] \big(\Psi_2^{-1/3} \Psi_{1B} \big) &=
-6\Psi_{2A}^{2/3}
[ ({\tho}'+2\rho'-\bar{\rho}')\kappa_B -({\edt}'+2\tau'-\bar{\tau})\sigma_B + 2\Psi_{1B}]
, \label{eq:psi1b-final} \\
 \big[ \squareTME_0 +8\Psi_2 \big] \big(\Psi_2^{-2/3} \Psi_{2B} \big) &=
- 3\square_B  \Psi_2^{1/3}
 \label{eq:psi2b-final} ,
\\
\big[\squareTME_{-2}-4\Psi_2\big] \big( \Psi_2^{-1} \Psi_{3B} \big) &=
-6
[ ({\tho}+2\rho-\bar{\rho})\kappa_B' -({\edt}+2\tau-\bar{\tau}')\sigma_B' +
  2\Psi_{3B}] \label{eq:psi3b-final}
\\
\big[ \squareTME_{-4} -16 \Psi_2 \big] \big( \Psi_2^{-4/3} \Psi_{4B} \big) &=0	, \label{psi4b}
\end{align}
\end{thm}
\begin{proof}
First we consider the equations for linearized the Weyl scalars $\Psi_{0B},
\Psi_{4B}$
with extreme spin
weights $s=2,-2$.
The linearized Bianchi identities \eqref{eq:ghpbianchia}$^*$ and
\eqref{eq:ghpbianchia}
read
\begin{align}
({\tho}'-\rho')\Psi_{0B} &= (\edt-4\tau)\Psi_{1B} +3\sigma_B \Psi_2 \label{psi0_1} ,\\
({\edt}'-\tau')\Psi_{0B} &= (\tho -4\rho)\Psi_{1B} +3\kappa_B \Psi_2 . \label{psi0_2}
\end{align}
Combining these identities as
$(\tho -4\rho -\bar{\rho})\eqref{psi0_1} - (\edt -4\tau
-\bar{\tau}')\eqref{psi0_2}$ gives
\begin{equation}\begin{aligned}
&\big[(\tho -4\rho -\bar{\rho})({\tho}'-\rho')- (\edt -4\tau -\bar{\tau}')({\edt}'-\tau')\big]\Psi_{0B} = \\
& \hspace{3cm} \big[(\tho -4\rho -\bar{\rho})(\edt-4\tau) - (\edt -4\tau -\bar{\tau}')(\tho -4\rho)\big]\Psi_{1B} \\
& \hspace{4cm} +3\big[(\tho -4\rho -\bar{\rho})\sigma_B - (\edt -4\tau -\bar{\tau}')\kappa_B \big]\Psi_2 .
\end{aligned}\end{equation}
The term involving $\Psi_{1B}$ on the RHS vanishes due to \eqref{advcomm}
with $a=4$.
The perturbed Ricci identity \eqref{eq:ghpeinsteinc} reads $(\tho-\rho
-\bar{\rho})\sigma_B -(\edt-\tau -\bar{\tau}')\kappa_B = \Psi_{0B}$ and from
\eqref{eq:ghpDbianchi} it follows, that also the $\Psi_2$ term on the RHS
reduces to $3\Psi_{0B}\Psi_2$. Recalling the form of $\squareTME_p$
gives equation \eqref{psi0b} for $\Psi_{0B}$. Equation \eqref{psi4b} for
$\Psi_{4B}$ follows from this after applying a prime and using
\eqref{eq:tmeprime}.

As we shall see, the corresponding wave equations governing the linearized
Weyl scalars $\Psi_{1B}, \Psi_{3B}$ with spin weights $s=1, -1$
do not decouple in the sense
that other perturbed quantities than the given linearized Weyl scalar
are involved. The linearized Bianchi identities \eqref{eq:ghpbianchib}$^*$ and \eqref{eq:ghpbianchib}$'$ read
\begin{equation}\begin{aligned}
({\tho}'-2\rho')\Psi_{1B} &= \bigg\{(\edt-3\tau)\Psi_2 \bigg\}_B  ,\\
-({\edt}'-2\tau')\Psi_{1B} &= -\bigg\{(\tho-3\rho)\Psi_2  \bigg\}_B ,
\end{aligned}\end{equation}
Multiplying both equations by $\Psi_2^{-1/3}$ and using the Leibniz rule
gives
\begin{align}
({\tho}'-\rho')(\Psi_2^{-1/3}\Psi_{1B}) &=
\frac{3}{2}
\bigg\{(\edt-2\tau)\Psi_2^{2/3} \bigg\}_B \label{psi1_1} ,\\
-({\edt}'-\tau')(\Psi_2^{-1/3}\Psi_{1B}) &=
-\frac{3}{2}
\bigg\{(\tho-2\rho)\Psi_2^{2/3}	 \bigg\}_B \label{psi1_2} .
\end{align}
Here we have made use of the fact that $\Psi_2^{-1/3}$ can be moved inside
the $\{\}_B$ brackets in view of
the background Bianchi identities \eqref{eq:ghpDbianchi}.

Combining the above identities as
$(\tho-2\rho-\bar{\rho})\eqref{psi1_1} +
(\edt-2\tau-\bar{\tau}')\eqref{psi1_2}$, we get
\begin{multline*}
\bigg[(\tho-2\rho-\bar{\rho})({\tho}'-\rho') - (\edt-2\tau-\bar{\tau}')({\edt}'-\tau')\bigg](\Psi_2^{-1/3}\Psi_{1B}) = \\
\frac{3}{2}
\bigg\{ \big[(\tho-2\rho-\bar{\rho})(\edt-2\tau) -
  (\edt-2\tau-\bar{\tau}')(\tho-2\rho) \big]\Psi_2^{2/3} \bigg\}_B .
\end{multline*}
which gives
\begin{equation} \label{eq:psi1b-intermediate}
\big[\squareTME_2 -4\Psi_2 \big] \big( \Psi_2^{-1/3} \Psi_{1B} \big) =
3
\bigg\{ \big[(\tho-2\rho-\bar{\rho})(\edt-2\tau) - (\edt-2\tau-\bar{\tau}')(\tho-2\rho) \big]\Psi_2^{2/3} \bigg\}_B ,
\end{equation}
using \eqref{eq:squaretme-restrict}.
Expanding the right hand side of equation
\eqref{eq:psi1b-intermediate},
leaving off a factor of $3$, we have
\begin{multline*}
\big[(\tho-2\rho-\bar{\rho})(\edt-2\tau) - (\edt-2\tau-\bar{\tau}')(\tho-2\rho) \big]_B\Psi_2^{2/3} + \\
+ \big[(\tho-2\rho-\bar{\rho})(\edt-2\tau) -
  (\edt-2\tau-\bar{\tau}')(\tho-2\rho) \big] \big[\Psi_{2}^{2/3} \big]_B.
\end{multline*}
The second term vanishes due to the commutation relation \eqref{advcomm} with $a=2$, but for the first term, we must use identites valid off the
background.
Since $\Psi_2$ is a $\{0,0\}$ quantity, we find
\begin{equation}\label{eq:psi1bgauge}\begin{aligned}
&[(\tho-2\rho-\bar{\rho})(\edt-2\tau) - (\edt-2\tau-\bar{\tau}')(\tho-2\rho)]_B\Psi_{2A}^{2/3} = \\
&= [-\bar{\tau}'\tho -\kappa{\tho}' +\bar{\rho}\edt +\sigma{\edt}' -2\tau\tho -2(\tho\tau) +\\
&\hspace{2cm} +2\rho\edt +2(\edt\rho) - (2\rho+\bar{\rho}) (\edt-2\tau) +(2\tau+\bar{\tau}') (\tho-2\rho)]_B\Psi_{2A}^{2/3} \\
&= [-2\kappa\rho' + 2\sigma \tau' -2(\tho\tau) +2(\edt\rho)+2\tau\bar{\rho} -2\bar{\tau}'\rho]_B\Psi_{2A}^{2/3} \\
&= -2\Psi_{2A}^{2/3}[ ({\tho}'+2\rho'-\bar{\rho}')\kappa_B -({\edt}'+2\tau'-\bar{\tau})\sigma_B + 2\Psi_{1B}] ,
\end{aligned}\end{equation}
where we used the commutation relation \eqref{eq:ghpcommb} in the first step,
background Bianchi identities in the second step and the Ricci identities
\eqref{eq:ghpeinsteina}$'$ and \eqref{eq:ghpeinsteina} in the last step.
Applying a prime and
making use of \eqref{eq:tmeprime} gives \eqref{eq:psi3b-final}.
Finally we consider the spin weight 0 linearized Weyl scalar $\Psi_{2B}$.
Recall that
$\Psi_{2}$ is non-vanishing in a type D spacetime.
For this reason,
it is convenient in the calculations to leave some expressions as
$\{\}_B$ brackets. As in the previous cases,
we start with linearized Bianchi identities. From
\eqref{eq:ghpbianchib}$'$ and \eqref{eq:ghpbianchib}$^*$$'$, we get
\begin{equation}\begin{aligned}
\bigg\{({\tho}'-3\rho')\Psi_2 \bigg\}_B	 &= (\edt-2\tau)\Psi_{3B} ,\\
-\bigg\{({\edt}'-3\tau')\Psi_2 \bigg\}_B  &= -(\tho-2\rho)\Psi_{3B}.
\end{aligned}\end{equation}
Multiplying both equations by $\Psi_2^{-2/3}$ and using the Leibniz rule
gives
rescaled equations
\begin{align}
\bigg\{3({\tho}'-\rho')\Psi_2^{1/3} \bigg\}_B &= \Psi_2^{-2/3}(\edt-2\tau)\Psi_{3B} \label{psi2_1} ,\\
-\bigg\{3({\edt}'-\tau')\Psi_2^{1/3} \bigg\}_B &= -\Psi_2^{-2/3}(\tho-2\rho)\Psi_{3B} \label{psi2_2} .
\end{align}
Here we used the fact that $\Psi_2^{-2/3}$ can be moved inside
the $\{\}_B$
brackets because of the background Bianchi identities \eqref{eq:ghpDbianchi}.

To find a wave equation for $\Psi_{2B}$, we consider the combination
$(\tho-\bar{\rho})\eqref{psi2_1} + (\edt-\bar{\tau}')\eqref{psi2_2}$, which
gives
\begin{align}
&(\tho-\bar{\rho})\bigg\{3({\tho}'-\rho')\Psi_2^{1/3} \bigg\}_B -(\edt-\bar{\tau}')\bigg\{3({\edt}'-\tau')\Psi_2^{1/3} \bigg\}_B = \nonumber\\
& \hspace{2cm} (\tho-\bar{\rho})\left(\Psi_2^{-2/3}(\edt-2\tau)\Psi_{3B}\right)- (\edt-\bar{\tau}')\left(\Psi_2^{-2/3}(\tho-2\rho)\Psi_{3B}\right) .
\end{align}
Using the identities $(\tho -\bar{\rho})(\Psi_2^{-2/3}\phi) = \Psi_2^{-2/3}
(\tho-2\rho-\bar{\rho})\phi$ and $(\edt -\bar{\tau}')(\Psi_2^{-2/3}\phi) =
\Psi_2^{-2/3}$ $(\edt -2\tau -\bar{\tau}')\phi$, which follow
from \eqref{eq:ghpDbianchi},
the RHS vanishes due to \eqref{advcomm} with $a=2$. The operators on
the LHS can be put into the $\{\}_B$ brackets because of the background
Bianchi identities  \eqref{eq:ghpDbianchi}. This gives the identity
\begin{equation}\label{eq:psi2b-intermediate}
\bigg\{\! \big[ \squareTME_0 +2\Psi_2 \big] \Psi_2^{1/3} \! \bigg\} _B \hspace{-2mm} = 0 .
\end{equation}
Expanding the $\{\}_B$ bracket in equation \eqref{eq:psi2b-intermediate}
gives
\begin{subequations} \begin{align}
0 &= 3\bigg\{ \big[ \squareTME_0 +2\Psi_2 \big] \Psi_2^{1/3} \bigg\} _B \\
&= \bigg[\square +2\Psi_2 \bigg] \left(\Psi_2^{-2/3}\Psi_{2B} \right) +3\bigg[\square_B +2\Psi_{2B}\bigg] \Psi_2^{1/3} \label{psi2bcross}\\
&= \bigg[\square +8\Psi_2 \bigg] \left(\Psi_2^{-2/3}\Psi_{2B} \right) +3\square_B  \Psi_2^{1/3} .\label{psi2bexpl}
\end{align}\end{subequations}
and hence \eqref{eq:psi2b-final}. This completes the proof.
\end{proof}

For future reference, we state the following equations which
  were used in the  proof of theorem \ref{thm:eq-lin-grav}
\begin{cor} \label{cor:pert-weyl}
\begin{subequations}\label{eq:cor-eq}
\begin{align}
\big[\squareTME_2 -4\Psi_2 \big] (\Psi_2^{-1/3} \Psi_{1B})  &=
3
\bigg\{ \big[(\tho-2\rho-\bar{\rho})(\edt-2\tau) -
  (\edt-2\tau-\bar{\tau}')(\tho-2\rho) \big]\Psi_2^{2/3} \bigg\}_B ,
\label{eq:psi1b-intermediate-cor} \\
\bigg\{\! \big[ \squareTME_0 +2\Psi_2 \big] \Psi_2^{1/3} \! \bigg\} _B	&= 0 ,
\\
\big[\squareTME_2 -4\Psi_2 \big] (\Psi_2^{-1} \Psi_{3B}) &=
3
\Psi_2^{-2/3} \bigg\{ \big[(\tho-2\rho-\bar{\rho})(\edt-2\tau) -
  (\edt-2\tau-\bar{\tau}')(\tho-2\rho) \big]\Psi_2^{2/3} \bigg\}_B' ,
\end{align}
\end{subequations}
\end{cor}
The rescaled Bianchi identity \eqref{psi1_2} has the
same form as the perturbed Maxwell equation \eqref{eq:ghpmaxwell} on a
charged type D background ($\phi_{1A}\neq 0$)
\begin{align}
\bigg\{(\tho-2\rho)\phi_1 \bigg\}_B &= ({\edt}'-\tau')\phi_{0B} .
\end{align}
Therefore the decoupled electromagnetic perturbation equations follow
immediately, and we have the following result.
\begin{cor}
In a charged type D background the Maxwell components $\phi_{iB}, i=0,1,2$
fulfill the equations
\begin{align}
\big[\squareTME_2 -4\Psi_2 \big] \phi_{0B} &= 2\bigg\{
\big[(\tho-2\rho-\bar{\rho})(\edt-2\tau) -
  (\edt-2\tau-\bar{\tau}')(\tho-2\rho) \big]\phi_1 \bigg\}_B  \label{phi0b} \\
&\bigg\{\big[(\tho-\rho-\bar{\rho})({\tho}'-2\rho')-
  (\edt-\tau-\bar{\tau}')({\edt}'-2\tau') \big]\phi_1\bigg\}_B
=0 \label{phi1b} \\
\big[\squareTME'_2-4\Psi_2\big] \phi_{2B} &= 2\bigg\{
\big[(\tho-2\rho-\bar{\rho})(\edt-2\tau) -
  (\edt-2\tau-\bar{\tau}')(\tho-2\rho) \big]'\phi_1 \bigg\}_B \label{phi2b}
\end{align}
\end{cor}
For a charged type D spacetime, the background Bianchi identities
\eqref{eq:ghpbianchi} include a term involving $\phi_1\bar\phi_1$
and hence the simple rescaling used above for vacuum type D backgrounds does not apply. Instead one can use the background Maxwell equations \eqref{eq:ghpmaxwell}, namely
$$
\tho \phi_1 = 2 \rho \phi_1, \quad \edt \phi_1 = 2 \tau \phi_1 ,
$$
as in \cite{fackerell:1982}. But as the $\Psi_2$ rescaling is
singular for the limit of flat background, the $\phi_1$ rescaling does not work for the limit of uncharged background.

In the special case of a test Maxwell field on an uncharged background,
rescaling by $\Psi_2$ becomes
possible and
the equation for $\phi_{1B}$ reduces to the Fackerell-Ipser equation
\cite{fackerell:ipser}
$$
\big[\square +2\Psi_2\big](\Psi_2^{-1/3} \phi_{1B}) = 0 ,
$$
while the equations for $\phi_{0B}$ and $\phi_{2B}$ become the spin $s=\pm1$ TME$_s$.

\section{Gauge source functions} \label{sec:gauge}
In this section we consider the equations for the gauge dependent quantities
$\Psi_{1B}, \Psi_{2B}, \Psi_{3B}$ in more detail.
\subsection{Gauge source functions for the Einstein equations}
In \cite{friedrich:1985}, Friedrich derived a frame based, symmetric hyperbolic
system for the Einstein-Yang-Mills system. We specialize to the vacuum case
and set the conformal factor $\Omega = 1$. Then, the result of
\cite{friedrich:1985} gives a symmetric hyperbolic system for a set of
unknowns consisting of a null tetrad (or spin frame),
the spin coefficients and the Weyl
spinor.

Starting from a system of equations involving the
tetrad and the curvature components as variables, he identified the gauge
source functions for this system. These are, letting
$(x^\mu)$ be coordinates on $M$,
$$
F^\mu = \square x^\mu ,
$$
and
$$
F_{\fraka \frakb} =
\nabla^a  (\nabla_a e_\fraka^b)	 e_\frakb^c g_{bc} ,
$$
where $(e_\fraka)$ is a null tetrad, cf.
\cite[equations (2.6), (2.13)]{friedrich:1985}, see also \cite{friedrich:1996}.

For the frame based hyperbolic system considered by Friedrich one may freely
specify the gauge source functions as functions of the spacetime coordinate,
the tetrad, the connection
coefficinets, and the Weyl tensor components, i.e.
\begin{align}
F^\mu &= F^\mu(x^\alpha, e_\fraka^\alpha, \Gamma^\fraka_{\frakb \frakc},
W_{\fraka \frakb \frakc \frakd}) , \\
F_{\fraka \frakb} &= F_{\fraka \frakb} ( x^\alpha, e_\fraka^\alpha, \Gamma^\fraka_{\frakb \frakc},
W_{\fraka \frakb \frakc \frakd}) ,
\end{align}
without changing the principal part of the resulting symmetric hyperbolic
system, see the discussion in \cite{friedrich:1996}, in particular
\cite[p. 1462]{friedrich:1996}.
As we are in a geometric situation where it is natural %on the one hand
to
adapt to a specific background geometry and to use a GHP weighted tetrad, it
is convenient to consider the following modified gauge source functions. They
differ from the expressions given by Friedrich by lower order terms, which do
not change the principal part of the resulting reduced system.

To define the coordinate gauge source function,
fix a background metric $\hat g_{ab}$ on $M$, with Levi-Civita derivative
$\hnabla_a$ and let
$V^a$ be the tension field for the identity map $(M, g_{ab}) \to (M, \hat
g_{ab})$ defined by
$$
V^a \xi_a = g^{cd} (\hnabla_c - \nabla_c )\xi_d ,
$$
holds for any 1-form $\xi_a$,
 see
\cite{pap:VM-local} for details. Then a gauge source function for the
coordinate degrees of freedom can be given by the equation
\begin{equation}\label{eq:coord-gauge-full}
F^\mu = V^\mu .
\end{equation}
Further,
let $\Theta_a$ be the weighted GHP covariant derivative and let $(e_\fraka)$
be a weighted GHP tetrad. A gauge source function for the tetrad degrees of
freedom can be given by the equation
\begin{equation}\label{eq:tetrad-gauge-full}
F_{\fraka \frakb} = \Theta^a  (\Theta_a e_\fraka^b)  e_\frakb^c g_{bc}	.
\end{equation}
These are the expressions which we shall consider below.

\subsubsection{Gauge source functions for the Linearized Einstein equations}
Consider linearized perturbations around vacuum spacetime
$(M,g_{ab})$. The work of Friedrich on hyperbolic reductions carries over
immediately to the linearized vacuum field equations. Thus we may consider
linearized frame based systems with unknowns consisting of the linearized
tetrad, the linearized spin coefficients and the linearized Weyl scalars.
The
reduced system is extracted by specifying linearized coordinate and frame
gauge source functions $F^\mu_B$, $F_{B \fraka \frakb}$ which may be
specified freely as
functions of the unknowns which are linear as functions of
$e_{B \fraka}^\alpha, \Gamma^\fraka_{B \frakb
  \frakc} , W_{B \fraka \frakb \frakc \frakd}$. Thus, with this restriction,
we may consider gauge
source functions
\begin{align*}
F_B^\mu &= F_B^\mu(x^\alpha, e_{B \fraka}^\alpha, \Gamma^\fraka_{B \frakb
  \frakc} , W_{B \fraka \frakb \frakc \frakd} ) , \\
F_{B \fraka \frakb}  &= F_{B \fraka \frakb}
(x^\alpha, e_{B \fraka}^\alpha, \Gamma^\fraka_{B \frakb
  \frakc} , W_{B \fraka \frakb \frakc \frakd} ) .
\end{align*}

\subsection{Linearized Weyl scalars and gauge}\label{sec:gauge-weyl}
In the following discussion it is in some steps convenient to
use a $\delta$ to denote first order
linearized fields. In particular,
$\delta g_{ab} = h_{ab}$ and denote the resulting perturbations in
geometric fields defined in terms of $g_{ab}$ by e.g. $\delta R_{ab}$.
We have
$$
\delta R_{ab} = - \half \square h_{ab} - R_a{}^c{}_b{}^d h_{cd} + \nabla_{(a}
v_{b)} ,
$$
where, letting $h = g^{ab} h_{ab}$, and working in a coordinate system
$(x^\mu)$,
$$
v^\alpha = \nabla_\beta h^{\beta\alpha} - \half \nabla^\alpha h =  g^{\beta\nu} \delta
\Gamma^\alpha_{\beta \nu} .
$$
The vector field $v^\alpha$ defined by this expression
is precisely the linearization of the tension field
$V^\alpha$
around
$g_{ab}$ (playing the role of the background metric $\hat g_{ab}$
above).
Thus $v^\alpha$ is
the appropriate coordinate gauge source function for
linearized perturbations $h_{\alpha\beta}$ of $g_{\alpha\beta}$
and the gauge condition which corresponds to
\eqref{eq:coord-gauge-full} is given by
$$
v^\alpha = F_B^\alpha .
$$
The standard harmonic gauge, also known as deDonder gauge,
with respect to the background metric is given by the condition $F^\alpha_B =
0$. In this gauge, the linearized Einstein equations in terms of the
linearized metric $h_{ab}$ take the form of a
wave equation
$$
\square h_{ab} + 2 R_a{}^c{}_b{}^d h_{cd} = 0 ,
$$
where $\square = \nabla^c \nabla_c$ is the covariant d'Alembertian.

Let $(M,g)$ be a vacuum type D spacetime. Working in a principal null
tetrad,
let $\psi = \Psi_2^{1/3}$.
Recall that $\Psi_2$ is of type $\{0,0\}$ and is thus a well-defined function
on spacetime. In notation used in this section, the condition that the right
hand side of \eqref{eq:psi2b-final} vanishes, i.e.,
$$
\square_B \Psi_2^{1/3} = 0 ,
$$
takes the form
$$
(\delta \square) \psi = 0 .
$$
We have
\begin{equation}\label{boxb2}
(\delta \square) \psi = - h^{ab} \nabla_a \nabla_b \psi - v^a \nabla_a \psi .
\end{equation}
We can view this equation as specifying part of the coordinate gauge degrees
of freedom. In the Schwarzschild case, working in a principal tetrad, $\psi$ is
real, and hence \eqref{boxb2} specifies one component of $v^a$. In the
general case, $\psi$ is complex, while $v^a$ is real. Taking the real and
imaginary parts of \eqref{boxb2} gives two real equations for $v^a$.

In order
to analyze this equation in the Kerr case,
it is convenient to calculate in a coordinate system
and tetrad which is non-singular on the horizon. A tetrad in the ingoing Kerr
coordinate system (also known as ingoing Eddington-Finkelstein coordinates)
was described by Teukolsky \cite[\S 5]{teukolsky:1973:I}. In this tetrad, the
components
$\nabla_{\fraka} \psi$ are non-vanishing on the horizon.

As mentioned above, it is compatible with the well-posedness of
the reduced field equations to allow the gauge source functions to depend on
the Weyl scalars. Thus we may also consider gauge conditions of the form
\begin{equation}\label{eq:mysticgauge}
-3\square_B \Psi_2^{1/3} - 6
\Psi_2^{1/3} \Psi_{2B} = 0 ,
\end{equation}
which leads to the wave equation
\begin{equation}\label{eq:fack-ip-psi}
(\square + 2\Psi_2) (\Psi_2^{-2/3} \Psi_{2B}) = 0 ,
\end{equation}
for $\Psi_{2B}$. Thus, in the gauge given by \eqref{eq:mysticgauge},
$\Psi_2^{-2/3} \Psi_{2B}$ satisfies the Fackerell-Ipser equation. This
substantiates the discussion in the work of Crossman and Fackerell
\cite{crossman:fackerell:1982,fackerell:1982}. In the Kerr case, calculation
shows that the gauge source function given by \eqref{eq:mysticgauge} will
have terms depending on $1/a$, and hence this gauge condition behaves in a
singular manner in the Schwarzschild limit. Equation \eqref{eq:fack-ip-psi}
is not known to be separable or admit a symmetry operator,
see however \cite[p. 617]{fackerell:1982}. This discussion shows that in the
rotating case, also
a generalized harmonic gauge condition leading to a homogenous wave equation
$$
\square (\Psi_2^{-2/3} \Psi_{2B} ) = 0,
$$
(which admits symmetry operators)
is compatible with a well-posed reduced system.

Next we consider the phantom gauge condition.
A calculation shows that
$$
\Theta^a \NPm^b \Theta_a \NPl_b
= 2 [ {\tho}' \kappa - {\edt}' \sigma - \bar{\rho}' \kappa + \bar{\tau}\sigma +
  \Psi_1 ] .
$$
Restricting to a type D background this expression vanishes, and the first
order linearization gives
\begin{equation}\label{eq:gaugesourcetetrad}
F_{B lm} := \half \bigg\{\Theta^a \NPm^b \Theta_a \NPl_b  \bigg\}_B
= [ ({\tho}'-\bar{\rho}')\kappa_B -({\edt}'-\bar{\tau})\sigma_B +
  \Psi_{1B} ] .
\end{equation}
Now we can write the equation for $\Psi_{1B}$ in the form
$$
\big[\squareTME_2 -4\Psi_2 \big] (\Psi_2^{-1/3}\Psi_{1B}) =
- 6 \Psi_{2A}^{2/3}
[ F_{B lm} + 2 \rho' \kappa_B
+ 2 \tau' \sigma_B + \Psi_{1B} ] .
$$
Thus, Chandrasekhar's phantom gauge condition can be written in terms of the
gauge source function as
$$
F_{B lm} = - 2\rho' \kappa_B - 2 \tau' \sigma_B - \Psi_{1B} .
$$
In view of the discussion above, this form of the gauge source function is
compatible with a hyperbolic system for the linearized Einstein equations.
The equation for $\Psi_{3B}$ can be handled along the same lines.

\section{Gauge invariant equations on Schwarzschild background}\label{sec:ginvschw}

Price \cite{price:1972:II} has shown that for linearized gravity on a
Schwarzschild background, $r^3\Im\Psi_{2B}$ describes odd parity perturbations. He
used the special coordinates of Newman and Penrose \cite[p.572]{newman:penrose:1962}
for the perturbed spacetime, which can be understood as a coordinate gauge. Price
expressed the perturbed spin coefficients, which occur in the $\Psi_{2B}$ equation,
in terms of perturbed metric coefficients. These coefficients are real in the odd
parity case and therefore cancel. Price then used the definition of $\Im \Psi_{2B}$
in terms of perturbed Riemann tensor components, to relate it to the perturbed metric
coefficients. He showed that these coincide up to  a time derivative, in Regge-Wheeler
(RW) gauge, with the RW variable $Q$ for odd parity perturbations.

Starting from \eqref{psi2bexpl},
\begin{align} \label{eq:fullpsi2b}
 \bigg[ \square + 8 \Psi_2 \bigg] \left( \Psi_2^{-2/3} \Psi_{2B} \right)
 = - 3 \square_B  \Psi_2^{1/3} ,
\end{align}
we rederive the known result, that the imaginary part satisfies the gauge invariant
RW equation~\cite{moncrief:1974}, which reduces to the result of Price in RW gauge.
We also show that the real part gives the gauge invariant Zerilli equation. Price
was not able to derive this equation, possibly due to his choice of special coordinates.

The conventions of Martel and Poisson~\cite{martel:poisson:2005} will be used
in this section. Indices $a,b,\dots$ for coordinates $t,r$ and $A,B,\dots$ for
coordinates $\theta,\varphi$. This notation differs from the convention used by
Regge-Wheeler and Zerilli in that $K_{\text{MP}} = K_{\text{RW}}-\tfrac{1}{2}l(l+1)G$.
Metric perturbations are denoted $g_{\mu \nu} = g_{A \mu \nu} + p_{\mu \nu}$, since
$h_{\mu \nu}$ is used for some spherical harmonic decomposed components. It should be
noted that we still use the signature $(+---)$ while Martel and Poisson use the signature
$(-+++)$. Using Schwarzschild coordinates $g_{a b} = \diag(f,-f^{-1},-r^2,-r^2 \sin^2
\theta)$, where $f = 1-\tfrac{2M}{r}$ and $\Psi_2 = - M/r^3$, the Kinnersley frame reads
$$
 l^a = \big( f^{-1} , 1 , 0 , 0 \big) , \quad
n^a = \frac{1}{2} \big( 1 , -f , 0 , 0 \big) , \quad
 m^a = \frac{1}{\sqrt{2}r} \big( 0 , 0 , 1 , \frac{i}{\sin\theta} \big) .
$$
With $ \psi = \Psi_2^{1/3} $, the RHS of \eqref{eq:fullpsi2b} reduces to
\begin{equation}\label{schwarzschildgauge}
\begin{aligned}
  \square_B \psi & \stackrel{\eqref{boxb2}}{=}
  - p^{\mu\nu} \nabla^2_{\mu\nu} \psi - ( \nabla^\mu p^\sigma{}_\mu
  - \frac{1}{2} \nabla^\sigma p_\mu{}^\mu ) \nabla_\sigma \psi \\
  & = \big( - p^{\mu r} \partial_\mu -( \partial_\mu p^{r \mu})
  - p^{r \rho} \Gamma_{\mu \rho}^\mu + g^{rr} ( \partial_r p_\mu{}^\mu) \big) \partial_r \psi .
\end{aligned}
\end{equation}
Further simplifications will occur in odd and even part, which are investigated in the next
sections.

To relate $ \Psi_{2B} $ to Regge-Wheeler and Zerilli variables we do a calculation analogous
to that of Price~\cite[Appendix D]{price:1972:II}. Starting from
\begin{equation}
\begin{aligned}
- 2 \Psi_{2B} =
& R_A{}_{\alpha \beta \gamma \delta} (l^\alpha n^\beta l^\gamma n^\delta)_B
- R_A{}_{\alpha \beta \gamma \delta} (l^\alpha n^\beta m^\gamma \bar{m}^\delta)_B \\
& + R_B{}_{\alpha \beta \gamma \delta} (l^\alpha n^\beta l^\gamma n^\delta)_A
- R_B{}_{\alpha \beta \gamma \delta} (l^\alpha n^\beta m^\gamma \bar{m}^\delta)_A ,
\end{aligned}
\end{equation}
we get
\begin{align} \label{psi2bschw}
- 2 \Psi_{2B} = R_B{}_{trtr} + \tfrac{\i}{r^2 \sin \theta} R_B{}_{rt\theta\phi}
+ \frac{4M}{r^3} p_{\alpha\beta} l^\alpha n^\beta.
\end{align}
The second term is purely imaginary and known from the calculations of Price.
The perturbed Riemann tensor is related to metric perturbations $p_{\mu\nu}$ via
\begin{align} \label{eq:pertriemann}
R_B{}_{\alpha \beta \gamma \delta} = \tfrac{1}{2}(p_{\beta\gamma;\alpha\delta}
+ p_{\alpha\delta;\beta\gamma} - p_{\alpha\gamma;\beta\delta} - p_{\beta\delta;\alpha\gamma}
+ R_A{}_{\alpha\sigma\gamma\delta}p_\beta^\sigma + R_A{}_{\sigma\beta\gamma\delta}p_\alpha^\sigma) .
\end{align}
With these relations one can check explicitly that $\Im \Psi_{2B}$ corresponds to odd parity
and $\Re \Psi_{2B}$ corresponds to even parity perturbations.

\subsection{Imaginary part and Regge Wheeler equation}
The odd parity metric perturbation can be expressed as \cite[eq. 5.1 - 3]{martel:poisson:2005}
$$
p_{ab}=0 , \quad
p_{aB} = \sum_{lm} h_a^{lm} X_B^{lm} , \quad
p_{AB}=\sum_{lm} h_2^{lm} X_{AB}^{lm} ,
$$
where $X_B^{lm}$ and $X_{AB}^{lm}$ are vector and tensor spherical harmonics. It follows
that $\operatorname{tr}p_{\mu \nu}=0$ and after a short calculation \eqref{schwarzschildgauge}
reduces to
\begin{align*}
\square_B \psi	= 0 .
\end{align*}
$\Im \Psi_{2B}$ corresponds to odd parity perturbations and $\Im$\eqref{eq:fullpsi2b} reduces
in this case to
\begin{align*}
\big[ \square + 8\Psi_2 \big] (\Psi_2^{-2/3} \Im \Psi_{2B}) = 0 .
\end{align*}
Introducing tortoise coordinates $r_*$ by $\partial_{r*} = f\partial_r$ cancels $2\Psi_2$
from the potential, a rescaling by $r$ cancels first order $\partial_r$ terms and we are
left with a gauge invariant Regge Wheeler equation
\begin{align*}
\left[ \partial_t^2 -\partial_{r_*}^2 + f\frac{l(l+1)}{r^2}
- f\frac{6M}{r^3} \right](r^3\Im \Psi_{2B})=0 .
\end{align*}
To relate this to the Regge Wheeler variable $Q$, we look at
\eqref{psi2bschw}. For odd parity perturbations it gives
$(R_B^{odd})_{trtr}=0=p^\text{odd}_{\alpha\beta} \NPl^\alpha \NPn^\beta $ and
with \eqref{eq:pertriemann},
$$
\tfrac{R_{Brt\theta\phi}^{odd}}{r^2\sin\theta} = \frac{l(l+1)}{2} \left[
  \left( \frac{h_t}{r^2}\right)_{,r}
  -\left(\frac{h_r}{r^2}\right)_{,t}\right],
$$
(for convenience we suppress spherical harmonics and the related indices). It follows
that just the imaginary part contributes to the perturbations. We now have
\begin{align*}
-r^3\Im\Psi_{2B,t} &= \frac{r^3}{4} l(l+1) \left[ \left( \frac{h_t}{r^2}\right)_{,r}
-\left(\frac{h_r}{r^2}\right)_{,t}\right]_{,t} \\
&= \frac{f(l+2)!}{4r(l-2)!}\left[h_r +\half h_{2,r} -\frac{1}{r} h_2 \right]
= \frac{(l+2)!}{4(l-2)!}Q^\text{odd} ,
\end{align*}
which is the gauge invariant variable of Moncrief~\cite{moncrief:1974} and in
RW gauge reduces to the result of Price \cite{price:1972:II}.

\subsection{Real part and Zerilli equation}
The even parity metric perturbations can be written \cite[eq. 4.1 - 3]{martel:poisson:2005}
\begin{equation} \label{eq:evenpar}
p_{ab}=\sum_{lm} h_{ab}^{lm} Y^{lm} ,\quad
p_{aB}=\sum_{lm} j_a^{lm} Y_B^{lm} , \quad
p_{AB}=r^2\sum_{lm} K^{lm} \Omega_{AB} Y^{lm} +G^{lm} Y_{AB}^{lm} ,
\end{equation}
where $Y^{lm},Y_B^{lm}$ and $Y_{AB}^{lm}$ denotes the even parity scalar, vector and tensor
spherical harmonics. $\Re \Psi_{2B}$ is not gauge invariant, but transforms as
\begin{align*}
 \Psi_{2B}  \to \Psi_{2B}-\frac{3M}{r^4} \xi^r .
\end{align*}
For some coordinate gauge transformation $x^\mu \to x^\mu +\xi^\mu$. Therefore we rewrite
the wave equation for the gauge invariant quantity, using Appendix \ref{a:ginv}
\begin{align} \label{gINVpsi}
\widetilde{\Psi}_{2B} := r^3 \Re \Psi_{2B} +\frac{3M}{2} \bigg(K +\frac{\lambda}{2} G \bigg) ,
\end{align}
where $\lambda=l(l+1)$. The $\Psi_{2B}$ equation \eqref{eq:fullpsi2b} for even parity perturbations reads
\begin{equation} \label{gINVzer}
\begin{aligned}
&\left[ \partial_t^2 -\partial_{r_*}^2 +f\frac{\lambda}{r^2} -f\frac{6M}{r^3} \right] \widetilde{\Psi}_{2B} =\\
&\hspace{6mm} 3Mfr\square_B r^{-1} +\frac{3M}{2} \left[ \partial_t^2 -\partial_{r_*}^2 +f\frac{\lambda}{r^2}
-f\frac{6M}{r^3} \right] \left(K +\frac{\lambda}{2}G \right)  .
\end{aligned}
\end{equation}
The perturbed wave operator term \eqref{schwarzschildgauge} does not vanish, but gives
\begin{align}
\square_B \frac{1}{r} = -\frac{1}{r^2}\left[ \frac{\lambda f}{r^2} j_r -\left(\frac{3Mf}{r^2}
+ \frac{f^2}{2}\partial_r \right)h_{rr} +\partial_t h_{tr} +\left(\frac{M}{fr^2}
-\half\partial_r\right)h_{tt} +f\partial_r K\right].
\end{align}
Expanding \eqref{psi2bschw} for even parity perturbations gives after some calculations
$R_{B rt\theta\phi}=0$ and
\begin{equation}
\begin{aligned}
\Re \Psi_{2B} =& M(f^{-1}h_{tt} -fh_{rr}) \\ &+\frac{r^3}{4}\bigg[\partial_r^2 h_{tt}
+ \partial_t^2 h_{rr} -\frac{m}{r^2}\partial_r(f^{-1}h_{tt} -fh_{rr})-\frac{2m}{r^2f}\partial_t h_{tr}
- 2\partial_{tr}^2 h_{tr}\bigg] .
\end{aligned}
\end{equation}
The whole equation is now given in terms of metric perturbations. To compare
the result to others, we express all metric components in terms of gauge
invariants (denote with a tilde according to the conventions of Martel and Poisson).
Using the equations of Appendix \ref{a:ginv}, $\widetilde{\Psi}_{2B}$ takes the form
\begin{align}
\widetilde{\Psi}_{2B} = \frac{\Lambda r}{4} \left[ \widetilde{K}
+ \frac{2f}{\Lambda} \left(f\widetilde{h}_{rr} -r\partial_r \widetilde{K} \right) \right] ,
\end{align}
where $\lambda = l(l+1) ,\, \mu=(l-1)(l+2)=\lambda-2 , \, \Lambda = \mu +6M/r
$. The RHS of \eqref{gINVzer} can also be expressed in terms of gauge invariant quantities
\begin{align}
\text{RHS} = -f\frac{6M}{r^3} \widetilde{\Psi}_{2B} +\frac{3Mf\lambda}{2r^2} \widetilde{K} .
\end{align}
The first term cancels the RW potential on the LHS! Now we rescale
\eqref{gINVzer} by $\Lambda^{-1}$ (which depends on $r$),
\begin{align}
\left[ \partial_t^2 -\partial_{r_*}^2 +f\frac{\lambda}{r^2} +\Lambda(\partial_{r_*}^2 \Lambda^{-1})\right]
(\Lambda^{-1} \widetilde{\Psi}_{2B} ) = \frac{3Mf\lambda}{2r^2} \Lambda^{-1}\widetilde{K}
- 2(\partial_{r_*} \Lambda^{-1}) (\partial_{r_*} \widetilde{\Psi}_{2B}) .
\end{align}
A straight forward but tedious calculation shows that the RHS can be written as
\begin{align}
\frac{3Mf\lambda}{2r^2} \Lambda^{-1}\widetilde{K} -2(\partial_{r_*} \Lambda^{-1})
( \partial_{r_*} \widetilde{\Psi}_{2B}) = f\frac{6M}{r^3}\frac{\lambda}{\Lambda} \Lambda^{-1}\widetilde{\Psi}_{2B} ,
\end{align}
the new potential term on the LHS  simplifies to
\begin{align}
\Lambda(\partial_{r_*}^2 \Lambda^{-1}) = \frac{12Mf}{\Lambda^2r^4} \left( \frac{6M^2}{r} +3M\mu -\mu r \right) ,
\end{align}
and we finally end up with
\begin{align}
\left[ \partial_t^2 -\partial_{r_*}^2 +\frac{12Mf}{\Lambda^2r^4}\left( \frac{6M^2}{r} +3\mu M
+ \frac{\mu^2 r}{2} + \frac{\mu^2 r^2}{6M}\left(\frac{\mu}{2} +1\right) \right)\right]
( \Lambda^{-1} \widetilde{\Psi}_{2B} ) = 0 .
\end{align}
This is the gauge invariant Zerilli equation. The relation to Moncrief's gauge invariant variable
is simply $Q^\text{even} = 4 \Lambda^{-1} \widetilde{\Psi}_{2B}$.

\appendix

\section{Even partity perturbations in Schwarzschild coordinates} \label{a:ginv}
We used equations of the unpublished appendix of \cite{martel:poisson:2005}, which is available as
\href{http://arxiv.org/abs/gr-qc/0502028}{arXiv:gr-qc/0502028}. 
For convenience we repeat the required results of appendix C.

The even-parity metric perturbations \eqref{eq:evenpar} transform under even parity gauge
$\xi_a= (\xi_t^{lm}Y^{lm}, \xi_r^{lm}Y^{lm}, \xi^{lm}Y_A^{lm})$ as
\begin{align*}
\delta h_{tt} &= -2 \frac{\partial}{\partial t} \xi_t + \frac{2Mf}{r^2} \xi_r, &&&
\delta h_{tr} &= -\frac{\partial}{\partial r} \xi_t
- \frac{\partial}{\partial t} \xi_r + \frac{2M}{r^2 f} \xi_t, \\
\delta h_{rr} &= -2 \frac{\partial}{\partial r} \xi_r
- \frac{2M}{r^2f} \xi_r, &&&
\delta j_t &= -\frac{\partial}{\partial t} \xi - \xi_t, \\
\delta j_r &= -\frac{\partial}{\partial r} \xi - \xi_r
+ \frac{2}{r} \xi, &&&
\delta K &= -\frac{2f}{r} \xi_r + \frac{\lambda}{r^2} \xi, \\
\delta G &= -\frac{2}{r^2} \xi.
\end{align*}
Martel and Poisson extracted the following gauge invariant quantities
\begin{align*}
\tilde{h}_{tt} &= h_{tt} - 2 \frac{\partial}{\partial t} j_t
+ \frac{2Mf}{r^2} j_r + r^2 \frac{\partial^2}{\partial t^2} G
- Mf \frac{\partial}{\partial r} G, \\
\tilde{h}_{tr} &= h_{tr} - \frac{\partial}{\partial r} j_t
- \frac{\partial}{\partial t} j_r + \frac{2M}{r^2 f} j_t
+ r^2 \frac{\partial^2}{\partial t \partial r} G
+ \frac{r-3M}{f} \frac{\partial}{\partial t} G, \\
\tilde{h}_{rr} &= h_{rr} - 2 \frac{\partial}{\partial r} j_r
- \frac{2M}{r^2 f} j_r + r^2 \frac{\partial^2}{\partial r^2} G
+ \frac{2r-3M}{f} \frac{\partial}{\partial r} G, \\
\tilde{K} &= K - \frac{2f}{r} j_r
+ rf \frac{\partial}{\partial r} G
+ \frac{\lambda}{2} G,
\end{align*}
and with these, the vacuum field equations are
\begin{align*}
0 =& -\frac{\partial^2}{\partial r^2} \tilde{K}
- \frac{3r-5M}{r^2 f} \frac{\partial}{\partial r} \tilde{K}
+ \frac{f}{r} \frac{\partial}{\partial r} \tilde{h}_{rr}
+ \frac{(\lambda+2)r + 4M}{2r^3} \tilde{h}_{rr}
+ \frac{\mu}{2r^2 f} \tilde{K}, \\
0 =& \frac{\partial^2}{\partial t \partial r} \tilde{K}
+ \frac{r-3M}{r^2 f} \frac{\partial}{\partial t} \tilde{K}
- \frac{f}{r} \frac{\partial}{\partial t} \tilde{h}_{rr}
- \frac{\lambda}{2r^2} \tilde{h}_{tr}, \\
0 =& -\frac{\partial^2}{\partial t^2} \tilde{K}
+ \frac{(r-M)f}{r^2} \frac{\partial}{\partial r} \tilde{K}
+ \frac{2f}{r} \frac{\partial}{\partial t} \tilde{h}_{tr}
- \frac{f}{r} \frac{\partial}{\partial r} \tilde{h}_{tt}
+ \frac{\lambda r + 4M}{2r^3} \tilde{h}_{tt}
- \frac{f^2}{r^2} \tilde{h}_{rr}
- \frac{\mu f}{2r^2} \tilde{K}, \\
0 =& \frac{\partial}{\partial t} \tilde{h}_{rr}
- \frac{\partial}{\partial r} \tilde{h}_{tr}
+ \frac{1}{f} \frac{\partial}{\partial t} \tilde{K}
- \frac{2M}{r^2 f} \tilde{h}_{tr}, \\
0 =& -\frac{\partial}{\partial t} \tilde{h}_{tr}
+ \frac{\partial}{\partial r} \tilde{h}_{tt}
- f \frac{\partial}{\partial r} \tilde{K}
- \frac{r-M}{r^2 f} \tilde{h}_{tt}
+ \frac{(r-M)f}{r^2} \tilde{h}_{rr},
\end{align*}
\begin{align*}
0 =& -\frac{\partial^2}{\partial t^2} \tilde{h}_{rr}
+ 2 \frac{\partial^2}{\partial t \partial r} \tilde{h}_{tr}
- \frac{\partial^2}{\partial r^2} \tilde{h}_{tt}
- \frac{1}{f} \frac{\partial^2}{\partial t^2} \tilde{K}
+ f \frac{\partial^2}{\partial r^2} \tilde{K} \\
& + \frac{2(r-M)}{r^2 f} \frac{\partial}{\partial t} \tilde{h}_{tr}
- \frac{r-3M}{r^2 f} \frac{\partial}{\partial r} \tilde{h}_{tt}
 - \frac{(r-M)f}{r^2} \frac{\partial}{\partial r} \tilde{h}_{rr}
+ \frac{2(r-M)}{r^2} \frac{\partial}{\partial r} \tilde{K} \\
&+ \frac{\lambda r^2-2(2+\lambda)Mr+4M^2}{2r^4 f^2} \tilde{h}_{tt}
- \frac{\lambda r^2-2\mu Mr-4M^2}{2r^4} \tilde{h}_{rr}, \\
0=& \frac{1}{f} \tilde{h}_{tt} - f \tilde{h}_{rr}.
\end{align*}

\subsection*{Acknowledgements} We are grateful to Ji{\v{r}}{\'{\i}}
Bi{\v{c}}{\'a}k, Pieter Blue, Helmut Friedrich, Jacek Jezierski, Jean-Philippe
Nicolas and Bernd Schmidt for helpful
discussions. One of the authors (S.A.) gratefully acknowledges the support of
the \emph{Studienstiftung des deutschen Volkes} and the \emph{Centre for Quantum Engineering and Space-Time Research (QUEST)}.

\newcommand{\prd}{Phys. Rev. D}

\bibliographystyle{abbrv}
\bibliography{lin}

\begin{thebibliography}{10}

\bibitem{andersson:blue:kerr}
L.~Andersson and P.~Blue.
\newblock Hidden symmetries and decay for the wave equation on the {Kerr}
  spacetime, 2009.
\newblock \href{http://arxiv.org/abs/0908.2265}{arXiv.org:0908.2265}.

\bibitem{pap:VM-local}
L.~Andersson and V.~Moncrief.
\newblock Elliptic-hyperbolic systems and the {E}instein equations.
\newblock {\em Ann. Henri Poincar\'e}, 4(1):1--34, 2003.

\bibitem{bicak:1979}
J.~Bi{\v{c}}{\'a}k.
\newblock On the theories of the interacting perturbations of the
  {R}eissner-{N}ordstr\"om black hole.
\newblock {\em Czechoslovak J. Phys. B}, 29(9):945--980, 1979.

\bibitem{bini:etal:2002}
D.~{Bini}, C.~{Cherubini}, and R.~T. {Jantzen}.
\newblock {Teukolsky Master Equation ---de Rham Wave Equation for the
  Gravitational and Electromagnetic Fields in Vacuum---}.
\newblock {\em Progress of Theoretical Physics}, 107:967--992, May 2002.

\bibitem{bini:etal:2003}
D.~{Bini}, C.~{Cherubini}, R.~T. {Jantzen}, and R.~{Ruffini}.
\newblock {De Rham wave equation for tensor valued p-forms}.
\newblock {\em International Journal of Modern Physics D}, 12:1363--1384, 2003.

\bibitem{blue:maxwell}
P.~Blue.
\newblock Decay of the {M}axwell field on the {S}chwarzschild manifold.
\newblock {\em J. Hyperbolic Differ. Equ.}, 5(4):807--856, 2008.

\bibitem{breuer:book}
R.~A. Breuer.
\newblock {\em Gravitational perturbation theory and synchrotron radiation}.
\newblock Springer-Verlag, Berlin, 1975.
\newblock Lecture Notes in Physics, Vol. 44.

\bibitem{Buchdahl:1958xv}
H.~A. Buchdahl.
\newblock {On the compatibility of relativistic wave equations for particles of
  higher spin in the presence of a gravitational field}.
\newblock {\em Nuovo Cim.}, 10:96--103, 1958.

\bibitem{carter:mclenaghan:1979}
B.~{Carter} and R.~G. {McLenaghan}.
\newblock {Generalized total angular momentum operator for the Dirac equation
  in curved space-time}.
\newblock {\em \prd}, 19:1093--1097, Feb. 1979.

\bibitem{chandra:book}
S.~Chandrasekhar.
\newblock {\em The mathematical theory of black holes}.
\newblock Oxford Classic Texts in the Physical Sciences. The Clarendon Press
  Oxford University Press, New York, 1998.
\newblock Reprint of the 1992 edition.

\bibitem{ChristodoulouKlainerman:LinearFields}
D.~Christodoulou and S.~Klainerman.
\newblock Asymptotic properties of linear field equations in {M}inkowski space.
\newblock {\em Comm. Pure Appl. Math.}, 43(2):137--199, 1990.

\bibitem{christo:klain:book}
D.~Christodoulou and S.~Klainerman.
\newblock {\em The global nonlinear stability of the {M}inkowski space},
  volume~41 of {\em Princeton Mathematical Series}.
\newblock Princeton University Press, Princeton, NJ, 1993.

\bibitem{collinson:1976}
C.~D. {Collinson}.
\newblock {On the Relationship between Killing Tensors and Killing-Yano
  Tensors}.
\newblock {\em International Journal of Theoretical Physics}, 15:311--314, May
  1976.

\bibitem{crossman:fackerell:1982}
R.~G. {Crossman} and E.~D. {Fackerell}.
\newblock {Electrovac perturbations of rotating black holes}.
\newblock In {C.~Edwards}, editor, {\em Gravitational Radiation, Collapsed
  Objects and Exact Solutions}, volume 124 of {\em Lecture Notes in Physics,
  Berlin Springer Verlag}, pages 459--468, 1980.

\bibitem{Dafermos:2008en}
M.~Dafermos and I.~Rodnianski.
\newblock {Lectures on black holes and linear waves}.
\newblock 2008.
\newblock \href{http://arxiv.org/abs/0811.0354}{arXiv.org:0811.0354}.

\bibitem{edgar:etal:2009}
S.~B. {Edgar}, A.~{G{\'o}mez-Lobo}, and J.~M. {Mart{\'{\i}}n-Garc{\'{\i}}a}.
\newblock {Petrov D vacuum spaces revisited: identities and invariant
  classification}.
\newblock {\em Classical and Quantum Gravity}, 26(10):105022--+, May 2009.

\bibitem{ehlers:1974}
J.~{Ehlers}.
\newblock {The geometry of the (modified) GHP-formalism}.
\newblock {\em Communications in Mathematical Physics}, 37:327--329, Dec. 1974.

\bibitem{fackerell:1982}
E.~D. Fackerell.
\newblock Techniques for linearized perturbations of {K}err-{N}ewman black
  holes.
\newblock In {\em Proceedings of the {S}econd {M}arcel {G}rossmann {M}eeting on
  {G}eneral {R}elativity, {P}art {A}, {B} ({T}rieste, 1979)}, pages 613--634,
  Amsterdam, 1982. North-Holland.

\bibitem{fackerell:ipser}
E.~D. {Fackerell} and J.~R. {Ipser}.
\newblock {Weak Electromagnetic Fields Around a Rotating Black Hole}.
\newblock {\em \prd}, 5:2455--2458, May 1972.

\bibitem{fayos:ferrando:jaen:1990}
F.~{Fayos}, J.~J. {Ferrando}, and X.~{Ja{\'e}n}.
\newblock {Electromagnetic and gravitational perturbation of type D
  space-times}.
\newblock {\em Journal of Mathematical Physics}, 31:410--415, Feb. 1990.

\bibitem{fernandes:lun:1997}
J.~F.~Q. {Fernandes} and A.~W.~C. {Lun}.
\newblock {Gauge invariant perturbations of black holes. II. Kerr space-time}.
\newblock {\em Journal of Mathematical Physics}, 38:330--349, Jan. 1997.

\bibitem{finster:etal:bams:2009}
F.~Finster, N.~Kamran, J.~Smoller, and S.-T. Yau.
\newblock Linear waves in the {K}err geometry: a mathematical voyage to black
  hole physics.
\newblock {\em Bull. Amer. Math. Soc. (N.S.)}, 46(4):635--659, 2009.

\bibitem{friedrich:1985}
H.~{Friedrich}.
\newblock {On the hyperbolicity of Einstein's and other gauge field equations}.
\newblock {\em Communications in Mathematical Physics}, 100:525--543, Dec.
  1985.

\bibitem{friedrich:1996}
H.~{Friedrich}.
\newblock {Hyperbolic reductions for Einstein's equations}.
\newblock {\em Classical and Quantum Gravity}, 13:1451--1469, June 1996.

\bibitem{GHP}
R.~{Geroch}, A.~{Held}, and R.~{Penrose}.
\newblock {A space-time calculus based on pairs of null directions}.
\newblock {\em Journal of Mathematical Physics}, 14:874--881, July 1973.

\bibitem{guven:1976}
R.~{G{\"u}ven}.
\newblock {Hertzian gravitational potentials for type D space-times}.
\newblock {\em Journal of Mathematical Physics}, 17:1315--1319, July 1976.

\bibitem{harnett:1990}
G.~{Harnett}.
\newblock {The GHP connection: a metric connection with torsion determined by a
  pair of null directions}.
\newblock {\em Classical and Quantum Gravity}, 7:1681--1705, Oct. 1990.

\bibitem{Hughes:2000pf}
S.~A. Hughes.
\newblock {Computing radiation from Kerr black holes: Generalization of the
  Sasaki-Nakamura equation}.
\newblock {\em Phys. Rev.}, D62:044029, 2000.

\bibitem{jeffryes:1984}
B.~P. {Jeffryes}.
\newblock {Space-times with two-index Killing spinors}.
\newblock {\em Royal Society of London Proceedings Series A}, 392:323--341,
  Apr. 1984.

\bibitem{jezierski:lukasik:2009}
J.~Jezierski and M.~{\L}ukasik.
\newblock Conformal {Y}ano-{K}illing tensors in {E}instein spacetimes.
\newblock {\em Rep. Math. Phys.}, 64(1-2):205--221, 2009.

\bibitem{jordan:ehlers:sachs:1961:II}
P.~Jordan, J.~Ehlers, and R.~K. Sachs.
\newblock Beitr\"age zur {T}heorie der reinen {G}ravitationsstrahlung.
  {S}trenge {L}\"osungen der {F}eldgleichungen der allgemeinen
  {R}elativit\"atstheorie. {II}.
\newblock {\em Akad. Wiss. Lit. Mainz Abh. Math.-Nat. Kl.}, 1961:1--62, 1961.

\bibitem{lindblad:rodnianski:global}
H.~Lindblad and I.~Rodnianski.
\newblock Global existence for the {E}instein vacuum equations in wave
  coordinates.
\newblock {\em Comm. Math. Phys.}, 256(1):43--110, 2005.

\bibitem{lun:fackerell:1975}
A.~W.~C. {Lun} and E.~D. {Fackerell}.
\newblock {A Master Equation for Perturbations to the Schwarzschild Geometry}.
\newblock {\em Lettere Al Nuovo Cimento}, 13:653--656, aug 1975.

\bibitem{martel:poisson:2005}
K.~Martel and E.~Poisson.
\newblock Gravitational perturbations of the {S}chwarzschild spacetime: a
  practical covariant and gauge-invariant formalism.
\newblock {\em Phys. Rev. D (3)}, 71(10):104003, 13, 2005.

\bibitem{moncrief:1974}
V.~{Moncrief}.
\newblock {Gravitational perturbations of spherically symmetric systems. I. The
  exterior problem}.
\newblock {\em Annals of Physics}, 88:323--342, Dec. 1974.

\bibitem{newman:penrose:1962}
E.~{Newman} and R.~{Penrose}.
\newblock {An Approach to Gravitational Radiation by a Method of Spin
  Coefficients}.
\newblock {\em Journal of Mathematical Physics}, 3:566--578, May 1962.

\bibitem{penrose:1965}
R.~{Penrose}.
\newblock {Zero Rest-Mass Fields Including Gravitation: Asymptotic Behaviour}.
\newblock {\em Royal Society of London Proceedings Series A}, 284:159--203,
  Feb. 1965.

\bibitem{penrose:1973}
R.~{Penrose}.
\newblock {Naked Singularities}.
\newblock In {D.~J.~Hegyi}, editor, {\em Sixth Texas Symposium on Relativistic
  Astrophysics}, volume 224 of {\em New York Academy Sciences Annals}, pages
  125--+, 1973.

\bibitem{penrose:1975}
R.~{Penrose}.
\newblock {Twistor theory - Its aims and achievements}.
\newblock In {\em Quantum gravity; Proceedings of the Oxford Symposium,
  Harwell, Berks., England, February 15, 16, 1974. (A76-11051 01-90) Oxford,
  Clarendon Press, 1975, p. 268-407.}, pages 268--407, 1975.

\bibitem{PR:I}
R.~Penrose and W.~Rindler.
\newblock {\em Spinors and space-time. {V}ol.\ 1}.
\newblock Cambridge Monographs on Mathematical Physics. Cambridge University
  Press, Cambridge, 1987.
\newblock Two-spinor calculus and relativistic fields.

\bibitem{PR:II}
R.~Penrose and W.~Rindler.
\newblock {\em Spinors and space-time. {V}ol. 2}.
\newblock Cambridge Monographs on Mathematical Physics. Cambridge University
  Press, Cambridge, second edition, 1988.
\newblock Spinor and twistor methods in space-time geometry.

\bibitem{price:1972:II}
R.~H. {Price}.
\newblock {Nonspherical Perturbations of Relativistic Gravitational Collapse.
  II. Integer-Spin, Zero-Rest-Mass Fields}.
\newblock {\em \prd}, 5:2439--2454, May 1972.

\bibitem{regge:wheeler:1957}
T.~{Regge} and J.~A. {Wheeler}.
\newblock {Stability of a Schwarzschild Singularity}.
\newblock {\em Physical Review}, 108:1063--1069, Nov. 1957.

\bibitem{ryan:1974}
M.~P. {Ryan}.
\newblock {Teukolsky equation and Penrose wave equation}.
\newblock {\em \prd}, 10:1736--1740, Sept. 1974.

\bibitem{exact:solutions:book}
H.~Stephani, D.~Kramer, M.~MacCallum, C.~Hoenselaers, and E.~Herlt.
\newblock {\em Exact solutions of {E}instein's field equations}.
\newblock Cambridge Monographs on Mathematical Physics. Cambridge University
  Press, Cambridge, second edition, 2003.

\bibitem{stewart:book}
J.~Stewart.
\newblock {\em Advanced general relativity}.
\newblock Cambridge Monographs on Mathematical Physics. Cambridge University
  Press, Cambridge, 1990.

\bibitem{stewart:walker:1974}
J.~M. {Stewart} and M.~{Walker}.
\newblock {Perturbations of space-times in general relativity}.
\newblock {\em Royal Society of London Proceedings Series A}, 341:49--74, Oct.
  1974.

\bibitem{teukolsky:1972}
S.~A. {Teukolsky}.
\newblock {Rotating Black Holes: Separable Wave Equations for Gravitational and
  Electromagnetic Perturbations}.
\newblock {\em Physical Review Letters}, 29:1114--1118, Oct. 1972.

\bibitem{teukolsky:1973:I}
S.~A. {Teukolsky}.
\newblock {Perturbations of a Rotating Black Hole. I. Fundamental Equations for
  Gravitational, Electromagnetic, and Neutrino-Field Perturbations}.
\newblock {\em Astrophysical J.}, 185:635--648, Oct. 1973.

\bibitem{walker:penrose:1970}
M.~{Walker} and R.~{Penrose}.
\newblock {On quadratic first integrals of the geodesic equations for type $\{$
  22$\}$ spacetimes}.
\newblock {\em Communications in Mathematical Physics}, 18:265--274, Dec. 1970.

\bibitem{whiting:1989}
B.~F. {Whiting}.
\newblock {Mode stability of the Kerr black hole}.
\newblock {\em Journal of Mathematical Physics}, 30:1301--1305, June 1989.

\bibitem{zerilli:1970}
F.~J. {Zerilli}.
\newblock {Effective Potential for Even-Parity Regge-Wheeler Gravitational
  Perturbation Equations}.
\newblock {\em Physical Review Letters}, 24:737--738, Mar. 1970.

\end{thebibliography}

%\end{thebibliography}
\end{document}